\documentclass[11pt,a4paper]{amsart}
\usepackage[left=3cm,bottom=3cm,right=3cm,top=3cm]{geometry}
\usepackage{amsmath} % Umgebugen für abgesetzte Formeln
\usepackage{graphicx}
\usepackage[margin=10pt,font=small,labelfont=bf,labelsep=endash]{caption}
\usepackage{mathrsfs}
\usepackage[utf8]{inputenc}
\usepackage{amssymb}
\usepackage{datetime}
\usepackage{color}
\numberwithin{equation}{section}
\newtheorem{defi}{Definition}[section]
\newtheorem{lemma}[defi]{Lemma}
\newtheorem{prop}[defi]{Proposition}
\newtheorem{cor}[defi]{Corollary}
\newtheorem{theorem}[defi]{Theorem}
\newtheorem{rem}[defi]{Remark}

\begin{document}
\title[Static Solutions of the EVM-System and Thin Shell Limit]{Existence of Static Solutions of the Einstein-Vlasov-Maxwell System\\and the Thin Shell Limit}
\author{Maximilian Thaller}
\address{Mathematical Sciences, Chalmers University of Technology and the University of Gothenburg, SE-412 96 G\"oteborg, Sweden}
\email{maxtha@chalmers.se}

\date{\today}
\maketitle

\begin{abstract}
In this article the static Einstein-Vlasov-Maxwell system is considered in spherical symmetry. This system describes an ensemble of charged particles interacting by general relativistic gravity and Coulomb forces. First, a proof for local existence of solutions around the center of symmetry is given. Then, by virtue of a perturbation argument, global existence is established for small particle charges. The method of proof yields solutions with matter quantities of bounded support - among other classes, shells of charged Vlasov matter. As a further result, the limit of infinitesimally thin shells as solution of the Einstein-Vlasov-Maxwell system is proven to exist for arbitrary values of the particle charge parameter. In this limit the inequality (\ref{a_b_ineq_charged_1}) obtained by Andr\'easson in \cite{a09} becomes sharp. However, in this limit the charge terms in the inequality are shown to tend to zero.
\end{abstract}

\tableofcontents
%\newpage 

\section{Introduction}

The Einstein-Vlasov-Maxwell system (EVM-system) is a model to describe an ensemble of charged particles whose motion is governed by gravity and interaction via Coulomb forces. Static regular solutions of this system describe particle configurations whose distribution is not changing while the individual particles are in motion. The aim of developing a mathematical understanding of static solutions of the Einstein-Vlasov system or the EVM-system is to provide models for astrophysical objects, as e.g.~galactic nebulae. In this context a number of questions on static solutions have been investigated. For the uncharged Einstein-Vlasov system, which can be used to describe galaxies or galaxy clusters, it has been proved in several different settings that static asymptotically flat solutions exist and that for massive particles the matter quantities of these solutions can have compact support \cite{aft15,akr14,akr11,r99,r93}. For massless particles the existence of static solutions with matter quantities of compact support has been established in \cite{aft16}. \par
Static solutions can be interpreted as equilibrium states of a kinetic system in the sense that the space-time and all matter quantities stay constant in the time evolution. Whether these static solutions are stable or unstable is still little understood, see \cite{ar06} for a numerical study. It has however been investigated how dense the matter in such a state can be concentrated. Bounds on the mass-to-radius ratio have been proven for an extensive class of spherically symmetric static objects covering kinetic and fluid models \cite{a08}, see the equations (\ref{a_b_ineq_background}) in Section \ref{sect_pre_mat} and equation (\ref{a_b_ineq_charged_1}) below. If the mass contained within a sphere is too large such that this inequality is not satisfied, the object cannot be static. There is numerical evidence \cite{ac14, ar06} that such overly concentrated objects tend to collapse. \par
A very high concentration of Vlasov matter can lead to trapped surfaces \cite{ar10}, i.e.~black hole formation \cite{a12}. The concentration parameter
\begin{equation} \label{def_conc_pram_0}
\Gamma := \sup_{r\in(0,\infty)} \frac{2m(r)}{r},
\end{equation}
where $r$ is the area radius and $m(r)$ the Hawking mass, can serve as an indicator if a trapped surface arises. The criterion would be that $\Gamma \to 1$. For the uncharged Einstein-Vlasov system it is known \cite{a08} that $\Gamma \leq 8/9$ for each static solution, i.e.~there is a gap between the highest possible value of $\Gamma$ and $1$. In particular an adiabatic transition, i.e.~a sequence of static solutions, from a regular static solution to a black hole is ruled out. In the charged case however the situation is different, no such gap is guaranteed by the corresponding inequality 
\begin{equation} \label{a_b_ineq_charged_1}
\sqrt{\frac{m_g(r)}{r}} \leq \frac 1 3 + \sqrt{\frac{1}{9} + \frac{q(r)^2}{r^2}}.
\end{equation}
Here $m_g(r)$ is a mass parameter involving some charge terms in addition to the Hawking mass, see the definition (\ref{def_g_mass}) below in Section \ref{sect_pre_mat} for details, and $q(r)$ denotes the charge contained in a ball of radius $r$. For charged, static dust solutions transitions to black holes have been studied in \cite{mh11}. These transitions are achieved by a sequence of solutions whose matter quantities are confined to a sphere of decreasing radial coordinate while the mass is kept constant, in a particular frame. \par
Little is known about static solutions of the EVM-system. Some light on the properties of static solutions in spherical symmetry has been shed by the numerical study \cite{aer09}. The following numerical observations of \cite{aer09} are investigated analytically in this article and theorems are proven that capture the observed behavior. Firstly the solutions resemble the solutions of the uncharged system if the particle charge parameter is not chosen too large. Shell and  {\em multishell} solutions have been observed. If the particle charge parameter is increased a critical value is encountered. If this critical value is exceeded no solutions with compactly supported matter quantities could be constructed. Secondly, a sequence of charged thin shell solutions has been constructed in whose limit the inequality (\ref{a_b_ineq_charged_1}) becomes sharp. In this limit the shells become arbitrarily thin and the ratio $Q/M$ of the total charge and the total mass approaches zero. In the uncharged case there exists an analogous limit which is analytically well understood \cite{a08, a072}. \par
Finally, we mention an interesting observation of \cite{aer09} which is not addressed analytically in this article but which makes the study of static solutions of the EVM-system even more interesting. In \cite{aer09} a different sequence of static solutions with a limit in which the inequality (\ref{a_b_ineq_charged_1}) becomes sharp has been constructed numerically. It was observed that in this limit $R=M=Q$ where $R$ is the radius of the support of the matter quantities, $M$ is the ADM mass of the solution, and $Q$ is the total charge. This limit is very interesting since it shows behavior of the EVM-system with no correspondence in the uncharged case and it can serve as candidate for an adiabatic black hole transition. Its analytical understanding remains an open problem.\par
The first result of this article is a local existence result. For the proof a contraction argument is applied. This is a standard technique to show local existence which has already been used in similar settings \cite{r93,rr92, aft15}. In these settings, however, the Einstein-Vlasov system can be reduced to one single integro-differential equation. In the context of the EVM-system, a function triple which satisfies a coupled system of integro-differential equations is needed to be considered. The method has been adapted to this function triple in this article. Then, as a next result, for small particle charges the existence of solutions with matter quantities of bounded support is shown. These solutions are regular for all radii and asymptotically flat. For this global existence result, we use a perturbation argument to show that up to a certain radius a static solution of the EVM-system converges to the uncharged solution with the same model parameters, as the particle charge goes to zero. This result allows to perturb uncharged solutions with matter quantities of bounded support to obtain charged solutions that again have matter quantities of bounded support and finite mass. \par
The existence of thin shell solutions of the EVM-system is the third result of this article. It is based on the insight that the charge terms become small close to the center of symmetry, independently of other characteristic quantities of the solutions, as e.g.~the concentration parameter $\Gamma$, defined in (\ref{def_conc_pram_0}). This smallness is established via a bootstrap argument. This insight allows to employ the methods developed in \cite{a072} to prove the existence of a sequence of static shell solutions with a fixed particle charge that converges to an infinitesimally thin shell. In other words the ratio $R_2/R_1$ of outer and inner radius of the matter shells in this sequence converges to $1$. At the same time the radius of the support of the matter quantities tends to zero in this limit. If a sequence of shell solutions of the EVM-system approaches an infinitesimally thin shell the inequality (\ref{a_b_ineq_charged_1}) becomes sharp. However, the charge-to-radius ratio vanishes as the shells become infinitesimally thin. In other words, in the limit of thin shells the static solutions of the charged system behave like the solutions with zero particle charge. In particular, in this limit there is no transition to a black hole comparable to \cite{mh11}. \par
The question of existence of static solutions of the EVM-system has already been addressed in \cite{n13}. However, the results of the present article go beyond the results of \cite{n13} at many places and various technical aspects of static solutions of the Einstein-Vlasov system, as well as peculiarities of the EVM-system are treated with greater care in this article.\par
%Finally it should be mentioned that in the numerical study \cite{aer09} a different type of static solutions seeming to saturate the Buchdahl-Andr\'easson inequality has been observed. It is a solution describing a shell of charged Vlasov matter that is extremal in the sense that radius, mass and charge are equal.
This article is organized as follows. First, in Section \ref{sect_preliminaries}, the EVM-system is introduced and some known properties about static solutions that this article relies on, are mentioned. In Section \ref{sec_loc} the proof of local existence and a continuation criterion is presented. In Section \ref{sec_glo} the existence of charged solutions with matter quantities of bounded support and finite mass is shown for small particle charges. In Section \ref{sec_qr_small} we show that for small radii the charge density can be controlled by $r^2$, the areal radius squared. In Section \ref{sect_thin_shells} this fact is exploited to prove that for a fixed particle charge parameter there exists a sequence of solutions that approach an infinitesimally thin shell.

\subsection*{Acknowledgments}
The authors likes to thank H\aa{}kan Andr\'easson for a series of helpful discussions, as well as for his comments on the manuscript. He would also like to thank Thomas B\"ackdahl for helpful discussions and further comments on the manuscript. Finally the helpful comments and corrections and the careful review of the manuscript by the anonymous referees is gratefully acknowledged.

\section{The Einstein-Vlasov-Maxwell system} \label{sect_preliminaries}

In this section we state the EVM-system in a spherically symmetric, static setting and introduce the relevant objects. The intention is mostly to fix notation, for a detailed derivation of the equations, see e.g.~\cite{lic}.\par
Let $\mathscr M$ be a four dimensional manifold equipped with the Schwarzschild coordinates $t\in\mathbb R$, $r\in [0,\infty)$, $\vartheta\in [0,\pi]$, $\varphi \in [0,2\pi)$ and the Lorentzian metric $g$ of signature $(-,+,+,+)$. For this metric $g$ we take the ansatz
\begin{equation}
g_{(t,r,\vartheta,\varphi)} = -e^{2\mu(r)}\mathrm d t^2 + e^{2\lambda(r)} \mathrm dr^2 + r^2\mathrm d\vartheta^2 + r^2\sin^2(\vartheta)\, \mathrm d\varphi^2
\end{equation}
to incorporate spherical symmetry and staticity. We define the mass shell
\begin{equation}
\mathscr P = \{(x,p) \in T\mathscr M\,:\,g_x(p,p) = -1,\, p \mathrm{\,future\,directed}\}.
\end{equation}
The mass shell is a submanifold of the tangent bundle $T\mathscr M$ of the space-time manifold $\mathscr M$. The tangent bundle $T\mathscr M$ of $\mathscr M$ can be equipped with the coordinates $(t,r,\vartheta,\varphi,p^{(t)}, p^{(r)}, p^{(\vartheta)}, p^{(\varphi)})$, where $(p^{(t)}, p^{(r)}, p^{(\vartheta)}, p^{(\varphi)}) \in \mathbb R^{4}$ are the canonical momenta to the coordinates $(t,r,\vartheta,\varphi)$. The symmetry suggests to describe the system in terms of the variables
\begin{equation}
r,\quad w=e^{\lambda(r)} p^{(r)}, \quad L = r^4\left(\left(p^{(\vartheta)}\right)^2 + \sin^2(\vartheta) \left(p^{(\varphi)}\right)^2\right).
\end{equation}
The physical interpretation of these variables is the following. The variable $r$ is the area radius, $w$ is the radial momentum and $L$ is the square of the angular momentum. The particles are described by the particle distribution function $f\in C^1(\mathscr P;\mathbb R_+)$ satisfying the Vlasov equation.  The integral of $f$ over a volume in the mass shell gives the number of particles in the corresponding space-time volume with momenta in the corresponding volume in momentum space. We assume that the particle distribution function $f$ is static and spherically symmetric, i.e.~$f=f(r,w,L)$, by slight abuse of notation. In spherical symmetry the electro-magnetic field is entirely described by the function $q\in C^1([0,\infty);\mathbb R_+)$ describing the charge contained in a ball of radius $r$. The Vlasov equation in terms of the variables $r$, $w$, $L$ reads
\begin{equation}
w \frac{\partial f}{\partial r}+ \left(\frac{L}{r^3} + q_0\frac{q(r)}{r^2}  e^{\mu(r) + \lambda(r)} - \mu'(r) \left(1 + w^2 + \frac{L}{r^2}\right) \right) \frac{\partial f}{\partial w}=0.
\end{equation}
For details and definitions, see \cite{lic}. The particles move along the characteristic curves of the Vlasov equation. So a function $f\in C^1(\mathscr P;\mathbb R_+)$ is a solution of the Vlasov equation if and only if it is constant along these characteristics. \par
The quantities $L$ and $E$ are conserved along the characteristics of the Vlasov equation, where $E$ is given by
\begin{equation} \label{def_cons_quan}
E(r,w,L) = e^{\mu(r)}\sqrt{1+w^2+\frac{L}{r^2}} - I_q(r), \quad I_q = q_0\int_0^r e^{(\mu+\lambda)(s)}\frac{q(s)}{s^2} \mathrm ds,
\end{equation}
cf.~\cite{lic} for further explanations. Any function of $E$ and $L$ is then a solution of the Vlasov equation. In this article we consider {\em polytropic} ansatz functions
\begin{equation} \label{art_ans_f}
f(r,w,L) = \left[1-\frac{E(r,w,L)}{E_0}\right]_+^k [L-L_0]_+^\ell
\end{equation}
with parameters
\begin{equation} \label{param_constraints}
k\geq 0, \quad L_0 > 0, \quad \ell \geq 0, \quad E_0>0.
\end{equation}
The constants $E_0$ and $L_0$ can be interpreted as cut-off values. They bound the energy from above and the angular momentum form below, respectively. Furthermore $[x]_+ = x$ if $x\geq 0$ and $[x]_+=0$ otherwise. In terms of the matter distribution function $f$ an energy-momentum tensor $T=T_{\mu\nu}\mathrm dx^\mu\mathrm dx^\nu$ can be defined. This energy momentum tensor is divergence free and satisfies the dominant energy condition \cite{a05}. \par
In spherical symmetry the Vlasov matter enters the Einstein equations solely through the three matter quantities $\varrho$, $p$, $p_T$ which can be interpreted as follows. The quantity $\varrho$ is the energy density of the particles, $p$ the radial pressure, and $p_T$ the transversal pressure, i.e.~the pressure tangential to spheres of fixed radius $r$. These matter quantities are given by
\begin{align}
\varrho(r) &= \frac{\pi}{r^2} \int_0^\infty \int_{-\infty}^\infty f(r,w,L) \sqrt{1+w^2+ \frac{L}{r^2}} \mathrm dw \mathrm dL, \\
p(r) &= \frac{\pi}{r^2} \int_0^\infty \int_{-\infty}^\infty f(r,w,L) \frac{w^2}{\sqrt{1+w^2+L/r^2}} \mathrm dw \mathrm dL, \\
p_T(r) &= \frac{\pi}{2r^4} \int_0^\infty \int_{-\infty}^\infty f(r,w,L) \frac{L}{\sqrt{1+w^2+L/r^2}} \mathrm dw \mathrm dL.
\end{align}
Furthermore we give a working definition of the charge density $\varrho_q$ by
\begin{equation}
\varrho_q(r) = q_0 e^{\lambda(r)}\frac{\pi}{r^2} \int_0^\infty \int_{-\infty}^\infty f(r,w,L)\,\mathrm dw\,\mathrm dL.
\end{equation}
Again, see \cite{lic} for more details and a derivation of these formulas. \par
We are now ready to state the spherically symmetric, static EVM-system. It reads
\begin{align}
e^{-2\lambda(r)}\left(2r\lambda'(r) -1\right)+1 &= 8\pi r^2\left(\varrho(r) + \frac{q^2(r)}{2r^4}\right), \label{eieq1}\\
e^{-2\lambda(r)}\left(2r\mu'(r) +1\right)-1 &= 8\pi r^2 \left(p(r) - \frac{q^2(r)}{2r^4}\right), \label{eieq2}\\
q'(r) &= 4\pi r^2 \varrho_q(r), \label{maxeq}
\end{align}
with the boundary conditions
\begin{align}
\mu(0) = \mu_c, \quad \lambda(0) = q(0) = 0, \label{bound_inner} \\
\lim_{r\to\infty} \mu(r) = \lim_{r\to\infty} \lambda(r) = 0. \label{bound_outer}
\end{align}
For more more information on the system and details to its derivation the reader may consult \cite{lic, n05, rr92, sz14}. We did not write down all non-trivial Einstein equations but if (\ref{eieq1}) and (\ref{eieq2}) are solved, this implies a solution for the other Einstein equations.
\begin{rem}
In the present coordinates a solution is asymptotically flat if $\mu(r), \lambda(r) \to 0$ as $r\to\infty$. However, the cutoff energy $E_0$ and the limit 
\begin{equation}
\mu_\infty := \lim_{r\to \infty}\mu(r)
\end{equation}
are not independent parameters. The boundary condition $\mu_\infty = 0$ together with the choice of a central value $\mu_c$ implies a certain value for $E_0$. Conversely, a choice of $\mu_c$ and $E_0$ will imply a value for $\mu_\infty$. However, a shift of $\mu_\infty$ by a constant corresponds merely to a scaling of the time coordinate $t$. In the remainder of this article we assume $E_0=1$ bearing in mind that the boundary conditions (\ref{bound_outer}) follow by a rescaling of the time coordinate, if necessary.
\end{rem}

\section{Characterization of the matter quantities} \label{sect_pre_mat}

The following lemma characterizes the dependency of the matter quantities on the metric, the charge, and the radial coordinate.

\begin{lemma} \label{lem_mat_1}
For fixed parameters
\begin{equation} \label{param_constraints_p}
k\geq 0, \quad L_0 > 0, \quad \ell \geq 0, \quad E_0=1,
\end{equation}
there exist functions $g, h, k: \mathbb R_+ \times \mathbb R \times \mathbb R_+ \to\mathbb R_+$ such that the matter quantities $\varrho$, $p$, and $\varrho_q$ can be written as
\begin{align}
\varrho(r) &= g(r,\mu(r),I_q[\mu,\lambda,q](r)), \label{eq_rho_fun} \\
p(r) &= h(r,\mu(r),I_q[\mu,\lambda,q](r)),  \label{eq_p_fun} \\
 \varrho_q(r) &= q_0e^{\lambda(r)} k(r,\mu(r),I_q[\mu,\lambda,q](r)),  \label{eq_rhoq_fun}
\end{align}
with the integral
\begin{equation} \label{def_op_iq}
I_q[\mu,\lambda,q](r) = q_0\int_0^r e^{(\mu+\lambda)(s)}\frac{q(s)}{s^2} \mathrm ds
\end{equation}
defined in (\ref{def_cons_quan}) as electro-magnetic contribution to the particle energy. The functions $g,h,k$ are continuously differentiable. The functions and their partial derivatives are increasing in the first and the third argument. All three functions, as well as their partial derivatives with respect to the first and the third argument, are non-increasing in the second argument whereas their partial derivatives with respect to the second argument are non-decreasing.
\end{lemma}

\begin{proof}
A straight forward calculation yields that the functions $g$, $h$, and $k$ can be realized by the expressions
\begin{align}
g(r,u,I) &= c_\ell r^{2\ell} \int_{\sqrt{1+L_0/r^2}}^{e^{-u}(1+I)} \left(1-\varepsilon e^{u}+I\right)^k \varepsilon^2 \left(\varepsilon^2-\left(1+\frac{L_0}{r^2}\right)\right)^{\ell+\frac 1 2} \mathrm d\varepsilon, \label{expr_fct_g} \\
h(r,u,I) &= \frac{c_\ell}{2\ell+3} r^{2\ell} \int_{\sqrt{1+L_0/r^2}}^{e^{-u}(1+I)} \left(1-\varepsilon e^{u}+I\right)^k \left(\varepsilon^2-\left(1+\frac{L_0}{r^2}\right)\right)^{\ell+\frac 3 2} \mathrm d\varepsilon, \label{expr_fct_h} \\
k(r,u,I) &= c_\ell r^{2\ell} \int_{\sqrt{1+L_0/r^2}}^{e^{-u}(1+I)} \left(1-\varepsilon e^{u}+I\right)^k \varepsilon \left(\varepsilon^2-\left(1+\frac{L_0}{r^2}\right)\right)^{\ell+\frac 1 2} \mathrm d\varepsilon, \label{expr_fct_k}
\end{align}
with the integration variable
\begin{equation}
\varepsilon = \sqrt{1+w^2+\frac{L}{r^2}}
\end{equation}
and the constant
\begin{equation}
c_\ell = \int_0^1 \frac{s^\ell}{\sqrt{1-s}}\mathrm ds.
\end{equation}
Here we understand $g(r,u,I) = h(r,u,I) = k(r,u,I) = 0$ if $e^{-u}(1+I) \leq \sqrt{1+L_0/r^2}$. The details can be found in \cite{lic}. We check that all partial derivatives exist and are continuous by proving that  the {\em left derivative} or the {\em right derivative} exists with respect to each argument and that it is continuous. To this end we first perform a change of variables in the integrals in (\ref{expr_fct_g}) -- (\ref{expr_fct_k}) given by $E=\varepsilon e^{\mu(r)} - I_q(r)$. Let $\alpha\in\mathbb R$, $\beta > 0$, $t > 0$, and consider the function
\begin{equation}
\xi_{\alpha,\beta}(t,I) = \int_{E=t-I}^\infty \phi(E) (E+I)^\alpha\left((E+I)^2-t^2\right)^\beta\mathrm dE,
\end{equation}
with the shorthand $\phi(E) = [1-E]_+^k$. With this notation one obtains
\begin{align}
g(r,u,I) &=  c_\ell r^{2\ell} e^{-(2\ell+4)u} \xi_{2,\ell+\frac 1 2}\left(e^u\sqrt{1+\frac{L_0}{r^2}},I\right), \label{rho_xi} \\
h(r,u,I) &= \frac{c_\ell}{2\ell+3} r^{2\ell}e^{-(2\ell+4)u}  \xi_{0,\ell+\frac 3 2}\left(e^u\sqrt{1+\frac{L_0}{r^2}},I\right), \label{p_xi} \\
k(r,u,I) &= c_\ell r^{2\ell}e^{-(2\ell+3)u} \xi_{1,\ell+\frac 1 2}\left(e^u\sqrt{1+\frac{L_0}{r^2}},I\right). \label{rhoq_xi}
\end{align}
We analyze $\xi_{\alpha,\beta}$ and prove first differentiability with respect to the first argument, by showing that the left derivative with respect to the first argument exists and that it is continuous. Let $t>0$. We consider for $0 < \Delta t < t$
\begin{align*}
&\frac{1}{\Delta t}\left(\xi_{\alpha,\beta}(t-\Delta t,I)-\xi_{\alpha,\beta}(t,I)\right) \\
& = \frac{1}{\Delta t}\int_{E=t-I-\Delta t}^{t-I} \phi\left(E\right) (E+I)^\alpha \left((E+I)^2-(t-\Delta t)^2\right)^{\beta} \mathrm dE \\
&\;\quad + \int_{E=t-I}^\infty \phi\left(E\right) (E+I)^\alpha \frac{1}{\Delta t}\\
&\qquad\qquad\qquad\times\left[\left((E+I)^2-(t-\Delta t)^2\right)^{\beta} - \left((E+I)^2 - t^2\right)^{\beta}\right] \mathrm dE.
\end{align*}
We observe that in the first integral $E\leq t-I$ due to the integration limit. This implies immediately  $(E+I)^\alpha\leq t^\alpha$ and $0 \leq \left((E+I)^2-(t-\Delta t)^2\right)^{\beta} \leq (\Delta t(2t-\Delta t))^\beta$. So for the first integral we obtain
\begin{multline}
0 \leq \frac{1}{\Delta t}\int_{E=t-I-\Delta t}^{t-I} \phi\left(E\right) (E+I)^\alpha \left((E+I)^2-(t-\Delta t)^2\right)^{\beta} \mathrm dE \\\leq t^\alpha (\Delta t (2t-\Delta t))^\beta\, \sup_{E\in (t-I-\Delta t, t-I)} \phi(E).
\end{multline}
By the fact that $\phi$ is bounded and that the integral is non-negative we see that the first term goes to zero as $\Delta t \to 0$ since $\beta >0$. So, applying the dominated convergence theorem to the remaining term, we obtain
\begin{equation} \label{left_der_xi_1}
\lim_{\Delta t \to 0} \frac{1}{\Delta t}\left(\xi_{\alpha,\beta}(t-\Delta t,I)-\xi_{\alpha,\beta}(t,I)\right) =-2t\beta  \xi_{\alpha,\beta-1}(t,I).
\end{equation}
The right hand side in (\ref{left_der_xi_1}), i.e.~the left derivative of $\xi_{\alpha,\beta}$ with respect to the first argument, clearly is continuous. This implies that $\xi_{\alpha,\beta}$ is differentiable with respect to the first argument and that
\begin{equation}
\frac{\partial \xi_{\alpha,\beta}}{\partial t}(t,I) =-2t\beta  \xi_{\alpha,\beta-1}(t,I).
\end{equation}
Next we consider for $0 < \Delta I < I$
\begin{align*}
&\frac{1}{\Delta I} \left(\xi_{\alpha,\beta}(t,I+\Delta I)-\xi_{\alpha,\beta}(t,I)\right) \\
& = \frac{1}{\Delta I} \int_{E=t-I-\Delta I}^{t-I} \phi\left(E\right) (E+I+\Delta I)^{\alpha} \left((E+I+\Delta I)^2-t^2\right)^{\beta} \mathrm dE \\
&\;\quad + \int_{E=t-I}^\infty \phi\left(E\right) \frac{1}{\Delta I} \Big[(E+I+\Delta I)^{\alpha} \left((E+I+\Delta I)^2-t^2\right)^{\beta} \\
&\qquad\qquad\qquad -(E+I)^{\alpha} \left((E+I)^2-t^2\right)^{\beta}\Big]\mathrm dE.
\end{align*}
Again, the first integral goes to zero, since $\phi$ is bounded, $(E+I+\Delta I)^{\alpha} \leq (t+\Delta I)^\alpha$, and $\left((E+I+\Delta I)^2-t^2\right)^{\beta} \leq (\Delta I(2t+\Delta I))^\beta$. The remaining part, by the same reasoning as above, gives
\begin{equation}
\frac{\partial \xi_{\alpha,\beta}}{\partial I}(t,I) =\alpha  \xi_{\alpha-1,\beta}(t,I) + 2\beta  \xi_{\alpha+1,\beta-1}(t,I).
\end{equation}
For the differentiability with respect to $r$, we see by inspection of the formulas (\ref{rho_xi}) -- (\ref{rhoq_xi}) that only the point $r=0$ needs to be discussed. However, since we assume $\ell \geq 0$ the point $r=0$ is regular. 
\end{proof}

\begin{rem}
If one assumes $L_0>0$, as done in Sections \ref{sec_qr_small} and \ref{sect_thin_shells} of this article, then Lemma \ref{lem_mat_1} holds also for $\ell > - \frac 1 2$. In fact the only place where $\ell \in(-\frac 1 2,0)$ is problematic is monotonicity and differentiability with respect to $r$, close to $r=0$. However, if $L_0>0$ we automatically have a vacuum region in some neighborhood of zero.
\end{rem}

\begin{rem} \label{rem_est_p_rho}
By inspection of the formulas (\ref{expr_fct_g}) and (\ref{expr_fct_h}) one realizes that for all $r\in[0,\infty)$ there holds
\begin{equation}
p(r) \leq \frac{1}{2\ell+3}\varrho(r),
\end{equation}
a fact that will be used later.
\end{rem}

In the analysis of this article, not the spherically symmetric static EVM-system (\ref{eieq1})--(\ref{maxeq}) is considered but a reduced version which is obtained by the method of characteristics using the ansatz (\ref{art_ans_f}) for the matter distribution function $f$. Since in this article we are working with this reduced system we call its solution a solution of the EVM-system. The way in that this term is used in this article is made precise in the following definition.

\begin{defi} \label{def_sol_evm}
The collection $\left(\mu, \lambda, q\right)_{\mu_c}$ is said to be a regular, static, asymptotically flat solution with central value $\mu_c$ of the spherically symmetric Einstein-Vlasov-Maxwell system if the equations
\begin{align}
e^{-2\lambda(r)} &= 1-\frac{8\pi}{r}\int_0^r s^2 g(s,\mu(s),I_q[\mu,\lambda,q](s)) \mathrm ds - \frac{1}{r}\int_0^r \frac{q^2(s)}{s^2} \mathrm ds, \label{red_evm_1}\\
\mu'(r) &= e^{2\lambda(r)}\Bigg(4\pi rh(r,\mu(r),I_q[\mu,\lambda,q](r)) +\frac{4\pi}{r^2}\int_0^r s^2 g(s,\mu(s),I_q[\mu,\lambda,q](s)) \mathrm ds \label{red_evm_2} \\
&\qquad\qquad-\frac{q^2(r)}{2r^3}+\frac{1}{2r^2}\int_0^r \frac{q^2(s)}{s^2}\mathrm ds \Bigg), \nonumber \\
q' &= 4\pi r^2 q_0 e^{\lambda(r)}k(r,\mu(r),I_q[\mu,\lambda,q](r)), \label{red_evm_3}
\end{align}
and the boundary conditions (\ref{bound_inner}) and (\ref{bound_outer}) are fulfilled, and all functions are bounded and $C^1$ everywhere.
\end{defi}

In this work, at various places it is important to keep track of vacuum and matter regions. To this end we define the characteristic function
\begin{equation} \label{def_char_fun}
\gamma(r) := \ln(1 + I_q(r)) - \mu(r)  - \frac 12 \ln\left(1+\frac{L_0}{r^2}\right).
\end{equation}
All matter quantities are zero at the radius $r$ if and only if $\gamma(r) \leq 0$. This can be seen from the formulas (\ref{expr_fct_g}) -- (\ref{expr_fct_k}), in particular the integral limits. Namely, if $\gamma(r) \leq 0$ we have $e^{-\mu(r)}(1+I_q(r)) \leq \sqrt{1+L_0/r^2}$. \par % due to the cutoff energy $E_0=1$. \par
As a first consequence we observe that if $L_0>0$ there will be a vacuum region $[0,R_1(\mu_c)]$ around the center as can be seen by inspecting (\ref{def_char_fun}). This radius is given by
\begin{equation} \label{eq_def_r1}
R_1(\mu_c) = \sqrt{\frac{L_0}{e^{-2\mu_c}-1}}.
\end{equation}
Observe that $R_1\to 0$ if $\mu_c\to-\infty$. Moreover the Tolman-Oppenheimer-Volkov equation (TOV-equation),
\begin{equation} \label{tov_eq}
p'(r) = -\mu'(r)(\varrho(r)+p(r))-\frac{2}{r}(p(r)-p_T(r))+\frac{q(r)q'(r)}{4\pi r^4},
\end{equation}
holds true. It is a consequence of the fact that the energy momentum tensor is divergence free. For a derivation see e.g.~\cite{lic}. \par
Furthermore, in the analysis presented in this article, inequalities linking radius, charge, and the mass of a spherically symmetric, static solution of the EVM-system will be of great use. These inequalities can be seen as generalizations of the well-known Buchdahl inequality.  To this end we define convenient mass parameters. The Hawking mass $m(r)$ is defined by
\begin{equation} \label{def_h_mass}
m(r) := 4\pi \int_0^r s^2\varrho(s) \mathrm ds.
\end{equation}
Next we define the mass parameter $m_\lambda(r)$ by
\begin{equation} \label{def_m_lambda}
m_\lambda(r) := m(r) + \frac{1}{2}\int_0^r \frac{q(s)^2}{s^2}\mathrm ds.
\end{equation}
Finally the gravitational mass $m_g(r)$ is defined by
\begin{equation} \label{def_g_mass}
m_g(r) := m_\lambda(r) + \frac{q(r)^2}{2r}.
\end{equation}
Now the aforementioned inequalities can be stated. If $q_0=0$, i.e.~we consider an uncharged solution of the Einstein-Vlasov system (\ref{eieq1})--(\ref{maxeq}) fulfilling regularity condition (\ref{bound_inner}) at the center, which exists on the interval $[0,R)$, where $R=\infty$ is of course possible, we have
\begin{equation} \label{a_b_ineq_background}
\sup_{r\in(0,R)} \frac{2m(r)}{r} \leq \frac 8 9.
\end{equation}
This result is proved in \cite{a08}. If $q_0 > 0$, a generalized form holds, namely
\begin{equation} \label{a_b_ineq_charged}
\sqrt{\frac{m_g(r)}{r}} \leq \frac{1}{3} + \sqrt{\frac{1}{9} + \frac{q^2(r)}{3 r^2}}
\end{equation}
for all $r\in(0,\infty)$, which has been derived in \cite{a09}. \par
To conclude this section we state the trivial yet useful identity
\begin{equation} \label{trivial}
e^{-2\lambda(r)} = 1 - \frac{2m_\lambda(r)}{r},
\end{equation}
which follows directly from the Einstein equations (\ref{eieq1})--(\ref{eieq2}).

\section{Local-in-$r$ Existence} \label{sec_loc}

Local existence will be proved by a contraction argument. To this end, for a constant $\delta >0$, we define a set $\mathcal C \subset \left(C^0([0,\delta];\mathbb R)\right)^2 \times C^1([0,\delta];\mathbb R)$ of function-triples and an operator $\mathcal T:\mathcal C \to \mathcal C$. Then we prove that this operator $\mathcal T$ is a contraction on $\mathcal C$ with respect to the norm $\|\cdot\|_{\mathcal C}$ given by
\begin{equation} \label{def_norm}
\|(u,v,w)\|_{\mathcal C} = \sup_{r\in[0,\delta]} |u(r)| + \sup_{r\in[0,\delta]} |v(r)| + \sup_{r\in[0,\delta]} |w(r)| + \sup_{r\in[0,\delta]} |w'(r)|.
\end{equation}
Banach's fixed point theorem then implies the existence of a fixed point $(u_\mathrm{f},v_\mathrm{f},w_\mathrm{f})$ of $\mathcal T$ in $\mathcal C$. The operator $\mathcal T$ will be defined in such a way that this fixed point gives rise to a solution $\left(\mu, \lambda, q\right)_{\mu_c}$ of the EVM-system in the sense of Definition \ref{def_sol_evm} that exists on the interval $[0,\delta]$. To be precise, we will have
\begin{equation} \label{obt_sol}
(\mu,\lambda,q)(r) = \left(u_{\mathrm f}(r), -\frac 12 \ln\left(v_{\mathrm f}(r)\right), \frac{r^2}{q_0}w'_{\mathrm f}(r)\sqrt{v_{\mathrm f}(r)}e^{-u_{\mathrm f}(r)}\right).
\end{equation}
Here we write $(\mu,\lambda,q)(r)$ for $(\mu(r),\lambda(r),q(r))$, a notation that we will adopt for the rest of this section. The central value $\mu_c$ and the limit $\delta$ of the interval are incorporated in the definitions of $\mathcal C$ and $\mathcal T$. \par
Theorem \ref{theo_loc_ex} below is phrased in a more general way. It yields the existence of a solution not only on a small interval $[0,\delta]$ but, if certain conditions are satisfied, it yields also the existence of a solution on the interval $[\mathring r, \mathring r+\delta]$ for $\mathring r>0$. The second case is relevant when proving the continuation criterion in Proposition \ref{prop_cont_crit} below.

\begin{theorem} \label{theo_loc_ex} (Local Existence)\\
Let $\mathring r \geq 0$ and $0 > \mathring \mu > -\infty$, $\mathring \lambda \geq 0$, $0 \leq \mathring I \leq \infty$ such that 
\begin{equation} \label{cond_ring_quan}
e^{-2\mathring \lambda} > \frac{1}{20},
\end{equation}
and if $\mathring r = 0$ let $\mathring \lambda = \mathring I = 0$. Let furthermore $k$, $\ell$, $L_0$ be chosen such that they satisfy (\ref{param_constraints_p}), and let $q_0 \geq 0$. There exists $\delta\in (0, 1]$ such that a solution $(\mu,\lambda,q)_{\mu_{c}}$ of the Einstein-Vlasov-Maxwell system in the sense of Definition \ref{def_sol_evm} exists on the interval $[\mathring r,\mathring r + \delta]$ with $\mu(\mathring r) = \mathring \mu$, $\lambda(\mathring r) = \mathring \lambda$, and $q(\mathring r)$ is given in equation (\ref{qring}).
\end{theorem}

\begin{rem}
Showing local existence by a contraction argument as it will be done in the proof of Theorem \ref{theo_loc_ex} below is an established technique which has in a similar context already been applied e.g.~in \cite{r93}. The particularity of the present situation is, however, that we consider an operator defined on a function triple satisfying additional equations instead of a single function and that one of these functions needs to be continuously differentiable instead of just continuous, i.e.~the derivative and not only the function itself has to be controlled.
\end{rem} 

\begin{proof}
Define the set $\mathcal C$ to be
\begin{multline} \label{def_c}
\mathcal C = \Big\{ (u,v,w) \in \left(C^0([\mathring r, \mathring r + \delta];\mathbb R)\right)^2 \times C^1([\mathring r, \mathring r + \delta];\mathbb R)\;: \\ \mathrm{conditions\,(\ref{cond_1})-(\ref{add_cond_v})\,hold,\,for\,all\,}r\in[\mathring r,\mathring r + \delta]\Big\}
\end{multline}
with the conditions
\begin{align}
&(u,v,w)(\mathring r) = \left(\mathring \mu, e^{-2\mathring \lambda}, \mathring I\right), \label{cond_1} \\
&0\leq w'(\mathring r) \leq q_0 \frac{e^{\mathring \mu-1}}{\sqrt{2}}, \label{cond_2} \\
& (u,v,w)(r) \in G, \quad G = \left[\mathring \mu-1,\mathring \mu+1\right] \times \left[\frac{1}{50},2 \right] \times \left[\mathring I, \mathring I+1\right], \label{def_g} \\
&\frac{8\pi}{r}\int_{\mathring r}^r s^2 g\left(s,u(s),w(s)\right) \mathrm ds + \frac{1}{q_0^2 r} \int_{\mathring r}^r s^2 (w'(s))^2 v(s) e^{-2u(s)} \mathrm ds \leq \frac{1}{50}. \label{cond_4} 
\end{align}
In addition, if $\mathring r = 0$ we also have the condition
\begin{equation} \label{add_cond_v}
1-v(r) \leq Cr^2, \qquad \mathrm{for\;all}\;r\in[0,\delta]
\end{equation} 
for a constant $C$ which depends on $\mathring \mu$ and $q_0$ but not on $\delta$. For the rest of this section we adopt $C$ as notation for a constant with these properties whose value may change from line to line. The operator $\mathcal T:\mathcal C \to \left(C^0([\mathring r, \mathring r + \delta];\mathbb R)\right)^2 \times C^1([\mathring r, \mathring r + \delta];\mathbb R)$ is defined to be
\begin{equation}
\mathcal T(u,v,w) = \left(\mathcal T_1(u,v,w), \mathcal T_2(u,v,w), \mathcal T_3(u,v,w)\right),
\end{equation}
with the components
\begin{align}
\mathcal T_1[u,v,w](r) &= \mathring \mu + \int_{\mathring r}^r \frac{1}{v(s)} \Bigg(4\pi s h\left(s,u(s), w(s)\right)  \\
&\qquad\qquad\quad+ \frac{1-v(s)}{2s}+\frac{(w'(s))^2s}{2q_0^2}v(s) e^{-2u(s)} \Bigg) \mathrm ds, \nonumber \\
\mathcal T_2[u,v,w](r) &= e^{-2\mathring \lambda}-\frac{8\pi}{r}\int_{\mathring r}^r s^2 g(s,u(s), w(s)) \mathrm ds \\
&\qquad\quad- \frac{1}{q_0^2 r} \int_{\mathring r}^r s^2 (w'(s))^2v(s) e^{-2u(s)} \mathrm ds, \nonumber \\
\mathcal T_3[u,v,w](r) &=\mathring I + q_0 \int_{\mathring r}^r \frac{e^{u(s)}}{s^2\sqrt{v(s)}} \left(\mathring q + 4\pi q_0 \int_{\mathring r}^s \frac{\sigma^2}{\sqrt{ v(\sigma)}} k(\sigma,u(\sigma),w(\sigma)) \mathrm d\sigma\right) \mathrm ds.
\end{align}
where
\begin{equation} \label{def_ring_q}
\mathring q = \frac{w'(\mathring r)}{q_0} \mathring r^2 e^{-(\mathring \lambda + \mathring \mu)}.
\end{equation}
Note that if $(\mu,\lambda,q)_{\mu_c}$ is a solution of the EVM-system in the sense of Definition \ref{def_sol_evm} then $\mathcal T[\mu,e^{-2\lambda},I_q](r) = (\mu, e^{-2\lambda}, I_q)(r)$ where $I_q$ is the integral defined in (\ref{def_op_iq}). \par
By a lengthy but straightforward argument one establishes
\begin{equation}
\forall (u,v,w)\in\mathcal C:\quad \mathcal T(u,v,w) \in \mathcal C
\end{equation}
and that there is a constant $C>0$, independent of $\delta$ such that for all $(u_1,v_1,w_1)$, $(u_2, v_2, w_2)\in \mathcal C$ we have
\begin{equation}
\left\| \mathcal T(u_1, v_1, w_1) - \mathcal T(u_2,v_2,w_2) \right\|_{\mathcal C} \leq C \delta \left\|\left(u_1 - u_2, v_1 - v_2, w_1 - w_2\right)\right\|_{\mathcal C}.
\end{equation}

By Banach's fix point theorem there exists a fix point $(u_\mathrm{f},v_\mathrm{f},w_\mathrm{f})$ of $\mathcal T$ in $\mathcal C$. Since $\mathcal T$ is an integral operator the fixed point $(u_f, v_f, w_f)$ will not only be continuous functions, but even $C^1$ functions. This implies that a solution of the EVM-system in the sense of Definition \ref{def_sol_evm} can be obtained via the relation (\ref{obt_sol}). Now we see as well that the obtained solution fulfills
\begin{equation}\label{qring}
q(\mathring r) = \frac{\mathring r ^2}{q_0} w'(\mathring r) e^{-(\mathring \mu+\mathring \lambda)}.
\end{equation}
\end{proof}

\begin{prop} (Continuation Criterion) \label{prop_cont_crit} \\
Let $0 \leq C \leq \frac{19}{40}$ be a constant and let $(\mu,\lambda,q)_{\mu_c}$ be a static solution of the EVM-system that exists at least on the interval $[0,R_c)$, $R_c>0$. Then there exists $\delta$ depending on the choice of $C$ and the parameters of the solution such that, if for all $r\in[0,R_c)$ we have $\frac{m_\lambda(r)}{r} \leq C$, then the solution can be extended at least to the interval $[0,R_c+\delta]$. 
\end{prop}

\begin{rem}
The condition $ C \leq \frac{19}{40}$ in the statement of Proposition \ref{prop_cont_crit} has been made for convenience and in view of its application in Lemma \ref{lem_q0s_qrl12} below. Any constant smaller than $\frac 1 2$ is possible. 
\end{rem}

\begin{proof}
We wish to deduce the existence of $\delta$ with the asserted properties from Theorem \ref{theo_loc_ex}.  To this end we have to show that the solution meets the assumptions of this theorem for all $\mathring r\in[0,R_c)$. In particular, if we take $\mathring r\in[0,R_c)$ and set $\mathring \mu=\mu(\mathring r)$, $\mathring \lambda = \lambda(\mathring r)$, and $\mathring I = I_q[\mu,\lambda,q](\mathring r)$, we need to show that the conditions (\ref{cond_ring_quan}) are satisfied. We have $e^{-2\mathring \lambda} > \frac{1}{20}$ by the choice of $C$ and since there generally holds (\ref{trivial}). \par
Now we show that $\mu(r)$ and $I_q[\mu,\lambda,q](r)$ stay bounded for all $r\in[0,R_c)$ if $\frac{m_\lambda(r)}{r} \leq C$ for all $r\in[0,R_c)$. We define the functions
\begin{equation}
n(r) := 4\pi r p(r) + \frac{m(r)}{r^2} - \frac{q(r)^2}{2r^3} + \frac{1}{2r^2} \int_0^r \frac{q(s)^2}{s^2}\mathrm ds
\end{equation}
and
\begin{equation}
f_q(r) := \frac{1}{2r^2}\int_0^r \frac{q(s)^2}{s^2}\mathrm ds - \frac{q(r)^2}{2r^3} = -\frac{\mathrm d}{\mathrm dr}\left(\frac{1}{2r}\int_0^r \frac{q(s)^2}{s^2}\mathrm ds\right).
\end{equation}
Observe that $n(r)$ is the bracket in the $\mu$-equation (\ref{red_evm_2}) and that $f_q(r)$ is exactly the terms in $n(r)$ that involve $q(r)$. We have
\begin{equation} \label{eq_est_mu_by_int}
\mu(r) \leq \mu_c + C \left( \left| 4\pi\int_0^r s p(s) \mathrm ds\right| + \left|\int_0^r\frac{m(s)}{s^2}\mathrm ds \right| + \left| \int_0^r f_q(s)\mathrm ds \right|\right)
\end{equation}
Since $-f_q$ is the $r$-derivative of the second term in $\frac{m_\lambda(r)}{r}$ and we have by assumption $\frac{m_\lambda(r)}{r} \leq \frac 1 2$ for all $r\in[0,R_c)$ we conclude that  
\begin{equation} \label{eq_fq}
0 \leq \left| \int_0^r f_q(s) \mathrm ds \right| < \frac 1 2.
\end{equation}
The functions $\mu$ and $I_q$ can diverge only as $r$ tends towards $R_c$. In our analysis we can thus assume that $r\in\left[\frac{R_c}{2},R_c\right)$. So, since by assumption $\frac{m(r)}{r}<\frac{1}{2}$ we have that $\frac{m(r)}{r^2}$ is bounded. Next we deal with the term $4\pi r p(r)$ in (\ref{eq_est_mu_by_int}). We have already observed that the only negative term in the equation (\ref{red_evm_2}) for $\mu'$ is bounded for all $r\in[0,R_c)$ (cf.~the estimate (\ref{eq_fq})). So, since, as already said, $\mu$ can only diverge as $r$ tends towards $R_c$ we deduce that 
\begin{equation}
4 \pi \int_0^{\frac{R_c}{2}} s p(s) \mathrm ds < \infty.
\end{equation}
Using Remark \ref{rem_est_p_rho} we have
\begin{align}
4 \pi \int_{\frac{R_c}{2}}^r s p(s) \mathrm ds & \leq \frac{4\pi}{2\ell + 3} \int_{\frac{R_c}{2}}^r s\varrho(s) \mathrm ds \leq \frac{2}{(2\ell+3)R_c} m(r) \\
&\leq \frac{2}{2\ell+3} \frac{m(r)}{r} \leq \frac{1}{2\ell + 3}. \nonumber
\end{align}
From (\ref{eq_est_mu_by_int}) we conclude that $\mu(r)$ stays bounded on $[0,R_c)$. By inspection of the definition of $m_g$, (\ref{def_g_mass}), it is clear that $m_g(r) > \frac{q(r)^2}{2r}$. Then the inequality (\ref{a_b_ineq_charged}) implies 
\begin{equation}\label{simp_est_qr}
\frac{q(r)}{r}\frac{1}{\sqrt{2}} \leq \frac 1 3 + \sqrt{\frac 1 9 + \frac{q(r)^2}{3r^2}}\quad\Leftrightarrow\quad \frac{q(r)}{r} \leq 2\sqrt{2}.
\end{equation}
It follows that $I_q(r)$ is bounded by inspecting its formula in equation (\ref{def_cons_quan}). By Lemma \ref{lem_mat_1} the matter quantities are given as continuous functions of $r$, $\mu$, and $I_q$, and thus are bounded on $[0,R_c)$, as well. This implies immediately that $n(r)$ is bounded. So we can apply Theorem \ref{theo_loc_ex} to conclude the assertion of the proposition.
\end{proof}

\section{Solutions with small particle charge} \label{sec_glo}

To prove global-in-$r$ existence of static solutions of the EVM-system equipped with a central value $\mu_c < 0$, we use a perturbation argument. To this end we consider a {\em background solution} $\left(\mu_0,\lambda_0,q_0\right)_{\mu_c}$, i.e.~a solution to the problem with the same parameter choices $k$, $\ell$, $E_0$, $L_0$, (cf.~the definition (\ref{art_ans_f}) of the ansatz function for the matter distribution function), and $\mu_c$, but $q_0=0$. All quantities belonging to this background solutions will be indexed with a ``$0$''. It is useful to summarize some known results on the background solution in a lemma. In this section we treat the massless case, too. Therefore the particle mass $m_p$ is introduced.

\begin{lemma} \label{char_back}
Let $m_p\in\{0,1\}$, $E_0=1$, and $k$, $\ell$, $L_0$ as in (\ref{param_constraints}). For each $\mu_c < 0$ there exists a static solution $\Psi^{\mu_c}_0=(\mu_0,e^{2\lambda_0},0)_{\mu_c}$ of the uncharged Einstein-Vlasov system that extends to the whole $r$-axis $[0,\infty)$.
\begin{enumerate}
\item If $m_p=1$ and the parameter $k$ satisfies additionally $k< \ell + 3/2$ there exists for each $\mu_c < 0$ a radius $0<r_{\mathrm{vac}}<\infty$ such that $[r_{\mathrm{vac}},\infty)$ is a vacuum region, i.e.~all matter quantities are zero and the metric is given by the Schwarzschild metric.
\item If $m_p\in\{0,1\}$, for each choice of model parameters $k$, $\ell$, $L_0$ where $L_0>0$ there exists $\mu_c^{\mathrm{shell}} < 0$ such that for all $\mu_c \leq \mu_c^{\mathrm{shell}}$ there exist $0<R_1(\mu_c)<R_2(\mu_c) <\infty$ such that $[0,R_1(\mu_c)]$ and $[R_2(\mu_c), \infty)$ are vacuum regions. If a solution of the Einstein-Vlasov-Maxwell system admits these radii, we call it a shell solution.
\end{enumerate}
\end{lemma}

\begin{proof}
Existence is shown in \cite{r93} which holds independently of the particle mass. Compact support in the massive case is shown for example in \cite{rr13}, and for the massless case in \cite{aft16}.
\end{proof}

\begin{rem}
The existence of a vacuum region of the background solution implies that there is an interval on the radial axis such that $\gamma_0(r) < 0$ on this interval. This fact is important for the setup of the following theorem.
\end{rem}

\begin{theorem} (Existence for small particle charge) \label{theo_only_ex}\\
Let $m_p \in\{0,1\}$, $\mu_c < 0$, and $\Psi_0^{\mu_c}$ be the background solution corresponding to the central value $\mu_c<0$ and parameters $E_0=1$, $k$, $\ell$, $L_0$ satisfying (\ref{param_constraints}) and additionally $k < \ell + 3/2$. If $m_p=0$ assume in addition $\mu_c\leq\mu_c^{\mathrm{shell}}$. Define 
\begin{equation}
R_{\mathrm{vac}} := \left\{\begin{array}{ll}R_2(\mu_c)&\quad\mathrm{if}\;m_p=0\\r_{\mathrm{vac}}&\quad\mathrm{if}\;m_p=1\end{array}\right.
\end{equation}
and let $\Delta > 0$ be sufficiently small such that $\gamma_0(R_{\mathrm{vac}}+\Delta) < 0$. 
Then there exists a constant $C(R_{\mathrm{vac}},\Delta)>0$, depending on $R_{\mathrm{vac}}$ and $\Delta$, such that for all $q_0\leq C(R_{\mathrm{vac}},\Delta)$ there exists a spherically symmetric, asymptotically flat, static solution $\Psi^{\mu_c}=(\mu,e^{2\lambda},\frac{q}{r^2})_{\mu_c}$ of the EVM-system with particle charge parameter $q_0$ whose matter quantities are supported on $[0,R_{\mathrm{vac}}+\Delta]$. The constant is given by $C(R_{\mathrm{vac}},\Delta) = \frac{d}{C}$ where $d$ is defined in (\ref{eq_def_d_muc}) below and $C$ is the constant in equation (\ref{const_is_given}) below.
\end{theorem}

\begin{proof} %----------------------------------------------------------
Let
\begin{equation} \label{eq_def_d_muc}
d := \min \left\{ \frac{1}{2}, |\gamma_0(R_{\mathrm{vac}}+\Delta)| \right\}
\end{equation}
where $\gamma_0$ is the characteristic function, defined in (\ref{def_char_fun}), of the background solution. Let $q_0 \leq C(R_{\mathrm{vac}},\Delta)$ where the constant $C(R_{\mathrm{vac}},\Delta) > 0$ will be determined later. By Theorem \ref{theo_loc_ex} there exists $R_c>0$ such that there exists a solution $\Psi^{\mu_c}=(\mu,e^{2\lambda},\frac{q}{r^2})_{\mu_c}$ of the EVM-system on the interval $[0,R_c)$. Let henceforth $R_c$ be the {\em maximal} interval of existence, i.e. the maximal interval such that the solution exists on $[0,R_c)$. We define
\begin{equation} \label{def_r_delta}
r_\Delta := \sup \left \{ r\in[0,R_c) \,: \|\Psi^{\mu_c}(r)-\Psi_0^{\mu_c}(r)\|_1 < d \right\},
\end{equation}
where the norm is given by
\begin{equation}
\left\|(x_1,\dots,x_n)\right\|_1 = \sum_{i=1}^n \left|x_i\right|.
\end{equation}
Note that for all $r\leq r_\Delta$ we have in particular, using (\ref{trivial}),
\begin{equation} \label{m0_minus_ml}
\left|\frac{1}{1-\frac{2m_0(r)}{r}} - \frac{1}{1-\frac{2m_\lambda(r)}{r}}\right| \leq d
\end{equation}
which by virtue of (\ref{a_b_ineq_background}) implies 
\begin{equation}
\frac{m_\lambda(r)}{r} \leq \frac{1}{2}\left(1-\frac{1}{9 + d}\right) \leq \frac 12 \frac{17}{19} = \frac{17}{38}.
\end{equation}
Proposition \ref{prop_cont_crit} implies then that $r_\Delta < R_c$. \par
In the next step of this proof we show that $r_\Delta \geq R_\mathrm{vac} + \Delta$. This is done by a contradiction argument. Assume $r_\Delta < R_{\mathrm{vac}}+\Delta$. Using (\ref{m0_minus_ml}) we obtain for $r\leq r_\Delta$
\begin{equation} \label{eq_est_iq}
I_q(r) = q_0 \int_0^r e^{(\mu+\lambda)(s)} \frac{q(s)}{s^2}\mathrm ds \leq q_0 \sqrt{9+d}\, e^{\mu_{\infty,0}+d}\,d \, r=:q_0 r C_I.
\end{equation}
A straight forward argument using Lemma \ref{lem_mat_1} yields
\begin{equation}  \label{est_q}
q(r) = 4\pi q_0 \int_{0}^r s^2 e^{\lambda(s)}  k(s,\mu(s),I_q(s))  \mathrm ds \leq q_0 C r^{3+2\ell}
\end{equation}
for all $r\in[0,r_\Delta]$. Furthermore using (\ref{trivial}) we obtain through a  straight forward argument for all $r\in[0,r_\Delta]$ the estimates
\begin{align} 
\left|e^{2\lambda_0(r)} - e^{2\lambda(r)} \right| &= \left|\frac{1}{1-\frac{2m_0(r)}{r}}-\frac{1}{1-\frac{2m(r)}{r}-\frac{1}{r}\int_0^r \frac{q^2(s)}{s^2} \mathrm ds}\right| \label{estimate2}  \\
&\leq q_0^2  C r^{4+4\ell} + C \int_0^r s \left|\varrho(s)-\varrho_0(s)\right|\mathrm ds \nonumber 
\end{align}
and
\begin{align}
|\mu(r) - \mu_0(r)| &\leq \int_0^r |\mu'(s) - \mu_0'(s)| \mathrm ds \label{est_mu} \\
&\leq q_0^2  C r^{3+4\ell} + C \int_0^r s (|\varrho(r)-\varrho_0(r)| + |p(r)-p_0(r)|) \mathrm ds. \nonumber
\end{align}

Using these preliminary estimates (\ref{eq_est_iq}) and (\ref{est_mu}) we now derive an upper bound for the difference of the matter quantities $\varrho$ and $p$, and $\varrho_0$ and $p_0$, respectively. First we define for $r\in[0,R_{\mathrm{vac}}+\Delta]$ the function
\begin{multline} \label{cgh}
C_{gh}(r) := \sup\{| \partial_u g(r,u,I)| + | \partial_I g(r,u,I)| + | \partial_u h(r,u,I)| + | \partial_I h(r,u,I)|\,:\,\\u\in[\mu_c-d,\mu_{\infty,0}+d], I\in[0,q_0(R_{\mathrm{vac}}+\Delta)C_I]\}.
\end{multline}
Lemma \ref{lem_mat_1} implies that 
\begin{equation}
\sup_{r\in[0,R_{\mathrm{vac}}+\Delta]}C_{gh}(r) < \infty. 
\end{equation}
Furthermore, the lemma implies that the supremum in (\ref{cgh}) will be attained if and only if we set $u = \mu_c-d$ and $I = q_0(R_{\mathrm{vac}}+\Delta)C_I$. By virtue of the mean value theorem we have
\begin{equation}
|\varrho(r)-\varrho_0(r)| + |p(r)-p_0(r)| \leq C_{gh}(r) (|\mu(r)-\mu_0(r)| + I_q(r)).
\end{equation}
for all $r\in [0,r_\Delta]$. Inserting (\ref{est_mu}) and (\ref{eq_est_iq}) we obtain
\begin{align}
&|\varrho(r)-\varrho_0(r)| + |p(r)-p_0(r)| \\
&\qquad \leq C_{gh}(r) \left(q_0^2 C r^{3+4\ell} + C \int_0^r s (|\varrho(r)-\varrho_0(r)| + |p(r)-p_0(r)|) \mathrm ds + q_0 r C_I \right) \nonumber \\
&\qquad \leq q_0C r + C \int_0^r s (|\varrho(r)-\varrho_0(r)| + |p(r)-p_0(r)|) \mathrm ds \nonumber
\end{align}
for all $r\in[0,r_\Delta]$. Since $q_0$ is chosen to be smaller than 1 we can assume $q_0^2\leq q_0$ and since $\ell > -1/2$ we can absorb $r^{2+4\ell}$ by the constant $C$. \par
By a Gr\"onwall argument, estimating $r\leq R_{\mathrm{vac}} + \Delta$ if necessary, we then obtain
\begin{equation} \label{bound_matter}
|p(s)-p_0(s)|+|\varrho(s)-\varrho_0(s)| \leq q_0 C.
\end{equation}
Together with (\ref{estimate2}) and (\ref{est_q}) this gives
\begin{equation} \label{bound_phi}
\|\Psi^{\mu_c}(r_\Delta)-\Psi_0^{\mu_c}(r_\Delta)\|_1 \leq q_0 C,
\end{equation}
where we recall that $C$ denotes a constant depending on $R_{\mathrm{vac}}+\Delta$, but not on $r_\Delta$. Now the desired contradiction is obtained if the bound $C(R_{\mathrm{vac}}, \Delta)$ of $q_0$ is smaller than $d/C$. We conclude $r_\Delta \geq R_{\mathrm{vac}}+\Delta$. \par 
Finally, in the third step of this proof we check that $|\gamma(r)-\gamma_0(r)| < d$ for all $r\in[0,R_{\mathrm{vac}}+\Delta]$ if $q_0$ is chosen sufficiently small. This is an immediate consequence of (\ref{eq_est_iq}),
\begin{align} 
|\gamma(r) - \gamma_0(r)| &\leq |\mu_0(r) - \mu(r)| + \ln\left(I_q(r)+1\right) \nonumber \\
& \leq |\mu_0(r) - \mu(r)| + I_q(r) \leq  q_0 C + q_0C_I (R_{\mathrm{vac}}+\Delta) \nonumber \\
& \leq q_0 C. \label{const_is_given}
\end{align}
Take now $d/C$ as the bound $C(R_{\mathrm{vac}}, \Delta)$ for $q_0$. This implies that at $r = R_{\mathrm{vac}} + \Delta$, and we have 
\begin{align}
\gamma(R_{\mathrm{vac}} + \Delta) &\leq \gamma_0(R_{\mathrm{vac}} + \Delta) + |\gamma(R_{\mathrm{vac}} + \Delta) - \gamma_0(R_{\mathrm{vac}} + \Delta)| \label{eq_bound_gammas} \\
&\leq \gamma_0(R_{\mathrm{vac}} + \Delta) + d = 0.\nonumber
\end{align}
Thus $\varrho(R_{\mathrm{vac}}+\Delta) = p(R_{\mathrm{vac}}+\Delta)=0$. Thus the metric can be continued with the Reissner-Nordstr\"{o}m metric and the matter quantities with constant zero. The solution is at least $C^1$, depending on the choice of $\ell$ the regularity can be higher. The regularity of the matter quantities follows by Lemma \ref{lem_mat_1} and the regularity of the metric functions and $q$ is apparent from the differential equations (\ref{red_evm_1}) -- (\ref{red_evm_3}).
\end{proof}

\section{Bounds on the charge density} \label{sec_qr_small}

In the remainder of this article  we prove that the {\em thin shell limit} for solutions of the EVM-system exists. This means, that for an arbitrary, fixed particle charge parameter $q_0 \geq 0$ there exists a solution for every central value $\mu_c < 0$ below a certain upper bound called $\mu_c^{\mathrm{shell}}$. The matter quantities of these solutions are supported on a shell and as $\mu_c\to-\infty$ this shell becomes infinitesimally thin. The central observation for this discussion is that the charge to mass ratio, $q(r)/m(r)$, goes to zero, as $\mu_c \to -\infty$, regardless of the choice of the particle charge parameter $q_0$. \par
In the remainder of this article we assume $L_0 > 0$, so in particular we have $R_1>0$ (cf.~(\ref{eq_def_r1})). The strategy in this section to obtain the result $q(r) \leq C_q r^2$ (Proposition \ref{lem_qr_r2}) is a bootstrap argument. We start with the bootstrap assumption $q(r)/r \leq 1/2$ on an interval $[0,r_*]$ and improve it later. At the same time the existence of the solution on the interval in question has to be guaranteed.

\begin{lemma} \label{lem_q0s_qrl12}
Let $q_0>0$ and $\mu_c<0$ be fixed. Further let $R_c>0$ be the maximal radius of existence, i.e.~$[0,R_c)$ is the maximal interval the solution exists on. Define
\begin{equation}
r_* := \sup\left\{r\in[0,R_c)\,:\,\frac{q(s)}{s}\leq \frac 1 2\,\mathrm{for}\,\mathrm{all}\,s\leq r\right\}.
\end{equation}
Then there exists $\delta > 0$ such that $r_*\leq R_c - \delta$.
\end{lemma}

\begin{proof}
We have $q(R_1)/R_1 = 0$. So by continuity $r_*>R_1$. We have
\begin{equation}
e^{-2\lambda(r)} = 1 + \frac{q(r)^2}{r^2}-\frac{2m_g(r)}{r}.
\end{equation}
We introduce the variable $a=q(r)/r$. The inequality (\ref{a_b_ineq_charged}) implies then
\begin{equation}
e^{-2\lambda(r)} \geq 1 + a^2 - 2 \left( \frac 1 3 + \sqrt{\frac{1}{9} + \frac{a^2}{3}}\right)^2.
\end{equation}
The right hand side is a function that is decreasing in $a$ for $a\in[0,\frac 1 2]$. So, since $a=q(r)/r\leq 1/2$ on the interval $[0,r_*]$, we can calculate the lower bound
\begin{equation} \label{eq_l_b_e2l_nnf}
e^{-2\lambda(r)} \geq \frac{1}{20}. 
\end{equation}
This implies $\frac{m_\lambda(r)}{r} \leq \frac{19}{40}$ for all $r\in[0,r_*]$. Then Proposition \ref{prop_cont_crit} implies the assertion.
\end{proof}

\begin{lemma} \label{lem_gamma_bound}
We have for all $r\in[R_1,r_*]$
\begin{equation}
 \gamma(r) \leq \frac{q_0\sqrt{5}}{\sqrt{L_0}} r + \frac{7}{2} \ln\left(\frac{r}{R_1}\right).
\end{equation}
\end{lemma}

\begin{proof}
Let $r\in[R_1,r_*]$. Recall the definition (\ref{def_char_fun}) of the function $\gamma(r)$,
\begin{equation*}
\gamma(r) = \ln\left(1+I_q(r)\right) - \mu(r) - \frac 1 2 \ln\left(1+\frac{L_0}{r^2}\right).
\end{equation*}
We calculate
\begin{equation} \label{eq_formula_gammap}
\gamma'(r) = \frac{\mathrm d}{\mathrm dr} \ln(1+I_q(r)) - e^{2\lambda(r)}\left(4\pi rp(r)+\frac{m_\lambda(r)}{r^2}-\frac{q(r)^2}{2r^3}\right) + \frac{1}{r}\frac{L_0}{r^2+L_0},
\end{equation}
where we used the formula (\ref{red_evm_2}) for $\mu'(r)$. From the definition of $\gamma(r)$ (cf.~equation (\ref{def_char_fun})) one sees that $\gamma(r)\geq0$ is equivalent to
\begin{equation}
e^{\mu(r)} \leq \frac{1}{\sqrt{r^2+L_0}}r(1+I_q(r)).
\end{equation}
So we have by $\frac{q(r)}{r}\leq \frac 1 2$ and the bound (\ref{eq_l_b_e2l_nnf})
\begin{align} \label{est_ddr_1piq}
\frac{\mathrm d}{\mathrm dr}\ln\left(1+I_q(r)\right) &= \frac{q_0}{1+I_q(r)} e^{(\mu+\lambda)(r)} \frac{q(r)}{r^2}\\
&\leq \frac{q_0}{\sqrt{r^2+L_0}}e^{\lambda(r)} \frac{q(r)}{r} \leq \frac{q_0\sqrt{5}}{\sqrt{L_0}}. \nonumber
\end{align}
In the second summand (\ref{eq_formula_gammap}) we can drop everything but the only positive term to obtain an estimate from above, i.e.
\begin{align}
&- e^{2\lambda(r)}\left(4\pi rp(r)+\frac{m(r)}{r^2}-\frac{q(r)^2}{2r^3}+\frac{1}{2r^2}\int_0^r\frac{q(s)^2}{s^2}\mathrm ds\right)\leq
e^{2\lambda(r)}\frac{q(r)^2}{2r^3} \leq  \frac{5}{2r}.
\end{align}
We conclude that for all $r\in[R_1,r_*]$
\begin{equation} \label{eq_a_bound_gp}
\gamma'(r) \leq  \frac{q_0\sqrt{5}}{\sqrt{L_0}} + \frac{5}{2r} + \frac{L_0}{(r^2+L_0)r} \leq  \frac{q_0\sqrt{5}}{\sqrt{L_0}} + \frac{7}{2r}.
\end{equation}
Thus we have for $r\in[R_1,r_*]$
\begin{equation}
\gamma(r) =\int_{R_1}^r \gamma'(s) \mathrm ds \leq \frac{q_0\sqrt{5}}{\sqrt{L_0}} r + \frac{7}{2} \ln\left(\frac{r}{R_1}\right),
\end{equation}
and the claim is shown.
\end{proof}

In the following part it is convenient to have the shorthand
\begin{equation}
\kappa:=\ell+k+\frac 3 2\geq 1.
\end{equation}

\begin{lemma} \label{lem_est_mat_c}
Let $\mu_c < -\ln(2)/2$, arbitrary and $D_{\mu_c} :=\{r\in[R_1,\min\{2R_1,r_*\}]\,:\,\gamma(r)\geq 0\}$. Define $z(r):=\varrho(r) - p(r) - 2p_T(r)$. There exist positive constants $C_\varrho^{\mathrm l}$, $C_p^{\mathrm l}$, $C_z^{\mathrm l}$, $C_{\varrho_q}^{\mathrm l}$, $C_\varrho^{\mathrm u}$, $C_p^{\mathrm u}$, $C_z^{\mathrm u}$, $C_{\varrho_q}^{\mathrm u}$ which are independent of the choice of $\mu_c$ such that for all $r\in D_{\mu_c}$ we have
\begin{align*}
C_\varrho^{\mathrm l} \frac{\gamma(r)^{\kappa}}{r^4} \leq& \varrho(r) \leq C_\varrho^{\mathrm u} \frac{\gamma(r)^{\kappa}}{r^4},\qquad C_p^{\mathrm l} \frac{\gamma(r)^{\kappa+1}}{r^4} \leq p(r) \leq C_p^{\mathrm u} \frac{\gamma(r)^{\kappa+1}}{r^4},\\
C_z^{\mathrm l} \frac{\gamma(r)^{\kappa}}{r^2} \leq& z(r) \leq C_z^{\mathrm u} \frac{\gamma(r)^{\kappa}}{r^2},\qquad C_{\varrho_q}^{\mathrm l} q_0 e^{\lambda(r)} \frac{\gamma(r)^{\kappa}}{r^3} \leq \varrho_q \leq C_{\varrho_q}^{\mathrm u} q_0 e^{\lambda(r)} \frac{\gamma(r)^{\kappa}}{r^3}.
\end{align*}
\end{lemma}

\begin{rem}
The important point of Lemma \ref{lem_est_mat_c} is that the constants $C_\varrho^{\mathrm l}$, $C_p^{\mathrm l}$, $C_z^{\mathrm l}$, $C_{\varrho_q}^{\mathrm l}$, $C_\varrho^{\mathrm u}$, $C_p^{\mathrm u}$, $C_z^{\mathrm u}$, $C_{\varrho_q}^{\mathrm u}$ are independent of the choice of $\mu_c$ which determines $R_1$ and thereby the interval $D_{\mu_c}$ on the $r$-axis. This means that powers of $r$ cannot be absorbed by these constants.
\end{rem}

\begin{proof}
By definition of the characteristic function (\ref{def_char_fun}) we have the identity $e^{-\mu(r)}(1 + I_q(r))=e^{\gamma(r)}\sqrt{1+\frac{L_0}{r^2}}$. We insert this identity in the upper limits in the integrals in the formulas (\ref{expr_fct_g}) -- (\ref{expr_fct_k}) of the matter quantities. Furthermore we perform a change of variables in the integrals in these formulas from $\varepsilon$ to $\alpha$, given by
\begin{equation}
\alpha = \frac{\varepsilon - \sqrt{1+\frac{L_0}{r^2}}}{\left(e^{\gamma(r)} -1\right)\sqrt{1+\frac{L_0}{r^2}}}.
\end{equation}
This yields the identities
\begin{align*}
\varepsilon &= \alpha\left(e^{\gamma(r)}-1\right)\sqrt{1+\frac{L_0}{r^2}} + \sqrt{1+\frac{L_0}{r^2}}, \\
\mathrm d\varepsilon &= \left(e^{\gamma(r)}-1\right)\sqrt{1+\frac{L_0}{r^2}} \mathrm d\alpha, \\
1-\varepsilon e^{\mu(r)}+I_q(r) &= (1+I_q(r))(1-\alpha) e^{-\gamma(r)} \left(e^{\gamma(r)} - 1\right), \\
\varepsilon^2 -\left(1+\frac{L_0}{r^2}\right) &= \left(\alpha\left(e^{\gamma(r)}-1\right)+2\right) \alpha \left(e^{\gamma(r)} -1\right) \left(1+\frac{L_0}{r^2}\right).
\end{align*}
One obtains
\begin{align}
\varrho(r) &= c_\ell r^{2\ell}\left(1+\frac{L_0}{r^2}\right)^{\ell + 2} \left(e^{\gamma(r)}-1\right)^{\ell+k+\frac 3 2} \mathcal I_{\ell+\frac 3 2,2}, \\
p(r) &= \frac{c_\ell}{2\ell+3} r^{2\ell} \left(1+\frac{L_0}{r^2}\right)^{\ell+2}\left(e^{\gamma(r)}-1\right)^{\ell+k+\frac 5 2} \mathcal I_{\ell+\frac 5 2,0}\\
\varrho_q(r) &= q_0 c_\ell e^{\lambda(r)} r^{2\ell} \left(1+\frac{L_0}{r^2}\right)^{\ell+\frac 3 2} \left(e^{\gamma(r)}-1\right)^{\ell + k +\frac 3 2} \mathcal I_{\ell+\frac 3 2,1},
\end{align}
with the shorthand
\begin{equation}
\mathcal I_{m,n} = \int_0^1 \left((1+I_q)(1-\alpha)e^{-\gamma}\right)^k \alpha^m \left(\alpha(e^\gamma-1)+1\right)^n(\alpha(e^\gamma-1)+2)^{\ell+\frac 1 2} \mathrm d\alpha.
\end{equation}
A straight forward calculation yields that $z(r)$ is given by
\begin{equation}
z(r) = \frac{2\pi}{\ell+1}r^{2\ell} \left(1+\frac{L_0}{r^2}\right)^{\ell+1}\left(e^{\gamma(r)}-1\right)^{\ell+k+\frac 3 2} \mathcal I_{\ell+\frac 1 2,0}.
\end{equation}
See \cite{lic} for details. Consider the integral $\mathcal I_{m,n}$. There are constants $C_1$ and $C_2$ independent of $\mu_c$ such that $0<C_1\leq \mathcal I_{m,n} \leq C_2 <\infty$ if $I_q$ and $\gamma$ are bounded. The functions $I_q$ and $\gamma$ are indeed bounded for $r\in D_{\mu_c}$ as we shall now see. By Lemma \ref{lem_gamma_bound} we have $\gamma(r) \leq 2q_0\sqrt{5}+\frac 7 2 \ln(2)$ since $R_1\leq \sqrt{L_0}$. For $I_q$ we first collect a few facts. Since $q(r)/r\leq 1/2$ we have the bound (\ref{eq_l_b_e2l_nnf}), i.e.~the equivalent bounds
\begin{equation} \label{equivalent_bounds}
e^{\lambda(r)} \leq \sqrt{20},\quad e^{2\lambda(r)}\leq 20,\quad \frac{m_\lambda(r)}{r}\leq \frac{19}{40}.
\end{equation}
For $\mu(r)$ we consider for $r\in[R_1,2R_1]$
\begin{equation}
\mu(r) = \mu_c + \int_{R_1}^r \mu'(s) \mathrm ds \leq \mu_c +\int_{R_1}^{2R_1}e^{2\lambda(s)}\left(4\pi s p(s) + \frac{m_\lambda(s)}{s}-\frac{q(s)^2}{2s^3}\right)\mathrm ds.
\end{equation}
Using $\mu_c \leq -\ln(2)$ by assumption, $p(r)\leq \frac{1}{2\ell+3}\varrho(r)$ by Remark \ref{rem_est_p_rho}, and $\frac{m_{\lambda}(r)}{r}<\frac 1 2$ we obtain 
\begin{align}
\mu(r) &\leq -\ln(2) + 20\left(\frac{4\pi}{R_1}\int_{R_1}^r s^2p(s) \mathrm ds + \frac5 8 \int_{R_1}^{r} \frac{1}{s}\mathrm ds \right) \nonumber \\
&=-\ln(2)+20\left(1+\frac{5}{8}\ln(2)\right) = 20-\frac 3 8 \ln(2), \nonumber
\end{align}
Now one can write down a bound for $I_q$ which is independent of the choice of $\mu_c$. We have
\begin{align}
I_q(r) &= q_0\int_{R_1}^r e^{(\mu+\lambda)(s)}\frac{q(s)}{s^2}\mathrm ds \\
&\leq C \int_{R_1}^r \frac{1}{s}\mathrm ds \leq C, \nonumber
\end{align}
with $C$ independent of $\mu_c$. \par
Next we consider the terms $\left(1+L_0/r^2\right)^\alpha$, where $\alpha$ is $\ell+2$, $\ell + 3/2$ or $\ell + 1$, respectively. We have
\begin{equation}
\left(1 + \frac{L_0}{r^2}\right)^\alpha = \frac{1}{r^{2\alpha}} \left(r + L_0\right)^\alpha.
\end{equation}
Now note that on $D_{\mu_c}$ we have $r \leq 2R_1 \leq 2\sqrt{L_0}$, and that we assume $L_0>0$. This means that there are constants $C$, $C'$ such that 
\begin{equation}
C \frac{1}{r^{2 \alpha}} \leq \left(1 + \frac{L_0}{r^2}\right)^\alpha \leq C' \frac{1}{r^{2\alpha}}.
\end{equation}
For the terms $\left(e^{\gamma(r)}-1\right)^\alpha$ with $\alpha = \ell + k + \frac 32$ or $\alpha = \ell + k + \frac 52$, respectively, we have, using that $0 \leq \gamma(r) \leq C$ for all $r\in D_{\mu_c}$, the estimates
\begin{equation}
C \gamma(r) \leq \left(e^{\gamma(r)}-1\right) \leq C' \gamma(r), \quad r\in D_{\mu_c},
\end{equation}
for two constants $C, C' > 0$, independent of the choice of $\mu_c$. One can see this for example by considering the Taylor expansion of the exponential function. The claim now follows in all four cases.
\end{proof}

\begin{prop} \label{lem_qr_r2}
There exists a constant $C_q>0$, independent of $\mu_c$, such that for all $r \in [0,\min\{r_*,2R_1\}]$ we have 
\begin{equation}
q(r)\leq C_q r^2.
\end{equation}
\end{prop}

\begin{proof}
Let $r \in [0,\min\{r_*,2R_1\}]$. We use the estimates for the matter quantities that we just proved. In particular we have \begin{equation} \label{eq_est_varrhoq_varrho}
\varrho_q(r) \leq q_0 \frac{C_{\varrho_q}^{\mathrm u}}{C_\varrho^{\mathrm l}} e^{\lambda(r)} r \varrho(r).
\end{equation}
As discussed above, $e^{\lambda(r)}$ is bounded, since $\frac{q(r)}{r} \leq \frac 1 2$ (cf.~equation  (\ref{eq_l_b_e2l_nnf})). We have
\begin{equation}
\frac{q(r)}{r} = \frac{4\pi}{r} \int_{R_1}^r s^2\varrho_q(s)\mathrm ds \leq \frac{q_0C}{r} \int_{R_1}^r s^3\varrho(s)\mathrm ds \leq q_0 C m(r) \leq q_0 C r,
\end{equation}
where in the last step $\frac{m(r)}{r}<\frac{1}{2}$ has been used again. Now, if we take the constant in the last term as $C_q$ the claim is established.
\end{proof}

\begin{cor} \label{cor_cor51}
If $R_1$ is sufficiently small then $2R_1 < r_*$.
\end{cor}
\begin{proof}
Choose $R_1<\frac{1}{2C_q}$ and assume that the corollary does not hold, i.e. $r_* \leq 2R_1$. Then for all $r\in[R_1,r_*]$ we have $\frac{q(r)}{r} \leq C_qr \leq C_qR_1 < \frac 1 2$, contradiction.
\end{proof}

\begin{cor} \label{cor_gamma_p}
We have for all $r\in[0,2R_1]$
\begin{equation}
\gamma'(r) \leq C r + \frac 1 r.
\end{equation}
\end{cor}

\begin{proof}
Carry out again the proof of Lemma \ref{lem_gamma_bound} but now use Proposition \ref{lem_qr_r2} instead of $\frac{q(r)}{r}\leq \frac 1 2$.
\end{proof}

\section{Limit of thin shells} \label{sect_thin_shells}

\begin{theorem} (Thin Shell Limit)\\
Let the model parameters $k$, $\ell$, $L_0$, $c_0$ be chosen as in (\ref{param_constraints_p}), with $L_0>0$, and let $q_0>0$ be arbitrary but fixed. Then there exists $\mu_c^{\mathrm{shell}}<0$ such that for all $\mu_c \leq \mu_c^{\mathrm{shell}}$ there exists a static solution of the Einstein-Vlasov-Maxwell system with particle charge $q_0$. The matter quantities of this solution are supported on $[R_1,R_2]$, where $R_2/R_1\to 1$, as $\mu_c\to-\infty$.
\end{theorem}

The proof of this theorem consists in a succession of lemmas. By inspecting the definition (\ref{eq_def_r1}) of $R_1(\mu_c)$ we see that $R_1(\mu_c)\to 0$ if and only if $\mu_c\to-\infty$. This means that the condition $\mu_c < 0$, $|\mu_c|$ large can be expressed saying that $R_1$ is small. We will adopt this terminology henceforth for convenience.

\begin{lemma} \label{lem_c_gamma}
Let
\begin{equation}
C_\gamma = (4\pi C_{p}^{\mathrm l})^{-\frac{1}{\kappa+1}}.
\end{equation}
Then, if $R_1$ is sufficiently small, for all $r\in[R_1,2R_1]$ there holds the bound 
\begin{equation}
\gamma(r) \leq C_\gamma r^{\frac{2}{\kappa+1}}.
\end{equation}
\end{lemma}

\begin{proof}
The proof works by contradiction. Assume that the assertion does not hold. Since we know that (by definition) $\gamma(R_1)=0$ we have $\gamma(R_1) < C_\gamma R_1^{\frac{2}{\kappa + 1}}$. This implies by continuity of $\gamma$ and $\gamma'$ that there exists $r_1\in (R_1,2R_1)$ such that 
\begin{equation}
\gamma(r_1) = C_\gamma r_1^{\frac{2}{\kappa+1}}
\end{equation}
and
\begin{equation}
\gamma'(r_1) > \frac{\mathrm d}{\mathrm dr} C_\gamma r^{\frac{2}{\kappa+1}} = \frac{2 C_\gamma}{\kappa +1} r^{\frac{1-\kappa}{\kappa+1}} > 0.
\end{equation}
Consider $\gamma'(r)$ given in (\ref{eq_formula_gammap}) that we recall here, 
\begin{multline*}
\gamma'(r) = \frac{\mathrm d}{\mathrm dr} \ln(1+I_q(r))\\
- e^{2\lambda(r)}\left(4\pi rp(r)+\frac{m(r)}{r^2}-\frac{q(r)^2}{2r^3}+\frac{1}{2r^2}\int_0^r\frac{q(s)^2}{s^2}\mathrm ds\right) + \frac{1}{r}\frac{L_0}{r^2+L_0}.
\end{multline*}
We now consider this expression for $r=r_1$ and estimate the positive terms by a negative one, one by one. First, since
\begin{equation}
\gamma(r_1)^{\kappa+1} = C_\gamma^{\kappa+1} r_1^2 = \frac{r_1^2}{4\pi C_p^{\mathrm{l}}},
\end{equation}
we have
\begin{equation} \label{lem_61_est1}
e^{2\lambda(r)} 4\pi r_1 p(r_1) \geq e^{2\lambda(r)} 4\pi C_p^{\mathrm l} \frac{\gamma(r_1)^{\kappa+1}}{r_1^3} = \frac{e^{2\lambda(r_1)}}{r_1} \geq \frac{1}{r_1} \frac{L_0}{r_1^2+L_0}.
\end{equation}
Next we claim that for $r\leq r_1$ we have $\frac{q(r)^2}{2r^3}\leq \frac{m(r)}{2r^2}$. To see this, we observe by inspection of the formulas (\ref{expr_fct_g}) and (\ref{expr_fct_k}) that 
\begin{equation} \label{eq_est_rho_rhoq}
\sqrt{1+\frac{L_0}{r^2}} \varrho_q(r) \leq q_0 e^{\lambda(r)}\varrho(r).
\end{equation}
This yields
\begin{equation} \label{est_qr_m}
\frac{q(r)}{r} = \frac{4\pi}{r}\int_{R_1}^r s^2\varrho_q(s)\mathrm ds \leq \frac{q_04\pi}{r}\int_{R_1}^r \frac{e^{\lambda(s)}}{\sqrt{1+\frac{L_0}{s^2}}} s^2 \varrho(s) \mathrm ds \leq \frac{q_0\sqrt{20}}{\sqrt{1+\frac{L_0}{r^2}}} \frac{m(r)}{r},
\end{equation}
where we use (\ref{equivalent_bounds}) to estimate $e^\lambda$. So if $R_1$ is small we have $q(r)\leq m(r)$. Then we have by Lemma \ref{lem_qr_r2}
\begin{equation} \label{lem_61_est2}
\frac{q^2(r)}{2r^3} \leq C_q \frac{m(r)}{2r} = C_q r \frac{m(r)}{r^2} \leq \frac{m(r)}{2r^2}
\end{equation}
The last inequality holds if $R_1$ is sufficiently small. For the next estimate, first recall the estimate (\ref{est_ddr_1piq}), which reads
\begin{equation*}
\frac{\mathrm d}{\mathrm dr}\ln\left(1+I_q(r)\right) \leq \frac{q_0}{\sqrt{r^2+L_0}} e^{\lambda(r)} \frac{q(r)}{r}.
\end{equation*}
Using $e^{\lambda(r)} \leq e^{2\lambda(r)}$, $q(r) \leq m(r)$ and $q_0/\sqrt{r^2+L_0}\leq \frac{1}{2r}$ for $R_1$ small we obtain
\begin{equation} \label{lem_61_est3}
\frac{\mathrm d}{\mathrm dr}\ln(1+I_q(r)) \leq e^{2\lambda(r)} \frac{m(r)}{2r^2}
\end{equation}
if $R_1$ is small enough. Combined, (\ref{lem_61_est1}), (\ref{lem_61_est2}), and (\ref{lem_61_est3}) imply $\gamma'(r_1)<0$ which is the desired contradiction.
\end{proof}

\begin{lemma}
Let
\begin{equation} \label{ans_delta}
\delta = \alpha R_1^{1+\beta},\qquad\mathrm{where}\quad \alpha=\left(64\pi C^{\mathrm u}\right)^{-\frac{1}{\kappa+1}}\quad\mathrm{and}\quad\beta=\frac{2}{\kappa+1},
\end{equation}
with $C^{\mathrm u} = \max\{C_\varrho^{\mathrm u}, 2 C_p^{\mathrm u}, C_{\varrho_q}^{\mathrm u}\}$. Then, if $R_1$ is chosen sufficiently small such that $\delta < R_1$, there holds the bound
\begin{equation}
\gamma'(r) \geq \frac{1}{2r}
\end{equation}
for all $r\in [R_1,R_1+\delta]$.
\end{lemma}

\begin{proof} We have
\begin{multline}
\gamma'(r) = \frac{q_0 e^{(\lambda+\mu)(r)}}{1+I_q(r)} \frac{q(r)}{r^2} \\
- e^{2\lambda(r)}\left(4\pi rp(r)+\frac{m(r)}{r^2}-\frac{q(r)^2}{2r^3}+\frac{1}{2r^2}\int_0^r\frac{q(s)^2}{s^2}\mathrm ds\right) + \frac{1}{r}\frac{L_0}{r^2+L_0}.
\end{multline}
In order to estimate $\gamma'(r)$ from below, we can drop the first term because it is positive. For $r\in[R_1,R_1+\delta]$ the last term can be estimated by
\begin{equation}
\frac 1 r \frac{L_0}{r^2 + L_0} \geq \frac 1 r \frac{L_0}{R_1^2 \left(1+\alpha R_1^\beta\right)^2 + L_0} \geq \frac{9}{10}
\end{equation}
for $R_1$ small enough. Now consider the middle term. We factor out $\frac 1 r$ and consider the remaining parenthesis,
\begin{equation} \label{rem_paran}
-e^{2\lambda(r)} \left(4\pi r^2p(r) + \frac{m(r)}{r}-\frac{q(r)^2}{2r^2}+\frac{1}{2r}\int_0^r\frac{q(s)^2}{s^2}\mathrm ds\right).
\end{equation}
The aim is now to find a lower bound for this expression. First we note that the term $e^{2\lambda(r)}\frac{q(r)^2}{2r^2}$ can be dropped because it is positive. \par
By Corollary \ref{cor_gamma_p} we have $\gamma'(r) \leq C r + \frac{1}{r}$. So we have for $r\leq 2R_1$
\begin{equation}
\gamma(r)^\kappa \leq \left((r-R_1) \max_{s\in[R_1,r]}\gamma'(s)\right)^\kappa \leq \left((r-R_1)\left(CR_1+\frac{1}{R_1}\right)\right)^\kappa.
\end{equation}
For $R_1$ small we can estimate $C R_1 + \frac{1}{R_1} \leq \frac{2}{R_1}$ and obtain 
\begin{equation}
\gamma(r)^\kappa \leq \frac{2}{R_1^\kappa} \delta^\kappa = 2\alpha^\kappa R_1^{\kappa\beta}.
\end{equation}
With this observation at hand we have for any $\sigma \leq \delta$ and $R_1$ sufficiently small
\begin{align*}
\frac{m(R_1+\sigma)}{R_1+\sigma} &\leq \frac{8\pi}{R_1+\sigma} C_\varrho^{\mathrm u} \int_{R_1}^{R_1+\sigma} \frac{1}{s^2} \alpha^\kappa R_1^{\kappa\beta} \mathrm ds \leq \frac{8\pi}{R_1+\sigma} C_\varrho^{\mathrm u} \alpha^{\kappa} R_1^{\kappa\beta} \left(\frac{1}{R_1}-\frac{1}{R_1+\sigma}\right)  \\
&\leq 8\pi C_\varrho^{\mathrm u} \alpha^{\kappa+1} R_1^{(\kappa+1)\beta-2}
\end{align*}
and
\begin{equation}
4\pi (R_1+\sigma)^2 p(R_1+\sigma) \leq 8\pi C_p^{\mathrm u} \frac{\sigma^{\kappa+1}}{R_1^{\kappa+1}(R_1+\sigma)^2} \leq 8\pi C_p^{\mathrm u} \alpha^{\kappa+1} R_1^{(\kappa+1)\beta-2}.
\end{equation}
With the choices for $\alpha$ and $\beta$ in (\ref{ans_delta}) we obtain for $r\in[R_1,R_1+\delta]$ the bounds $\frac{m(r)}{r}\leq \frac 1 8$ and $4\pi r^2 p(r)\leq \frac{1}{16}$. \par
Next we consider
\begin{align}
\frac{1}{2(R_1+\sigma)} \int_0^{R_1+\sigma} \frac{q(s)^2}{s^2}\mathrm ds &\leq \frac{1}{2R_1}\frac{C_q}{3}\left( \left(R_1+\sigma \right)^3 -R_1^3\right)\leq CR_1^{\beta+2}
\end{align}
where we have used Proposition \ref{lem_qr_r2}. So this term goes to zero, as $R_1\to 0$, and for $R_1$ sufficiently small we can assume $\frac{1}{2r}\int_0^r\frac{q^2}{s^2}\mathrm ds \leq \frac{1}{16}$ for $r\leq R_1+\delta$. This implies
\begin{equation}
e^{2\lambda(R_1+\sigma)} = \left(1-\frac{2m(R_1+\sigma)}{R_1+\sigma} - \frac{1}{R_1+\sigma}\int_0^{R_1+\sigma}\frac{q(s)^2}{s^2}\mathrm ds\right)^{-1} \leq \frac{8}{5}.
\end{equation}
If we insert these upper bounds into (\ref{rem_paran}) we obtain for $r\in[R_1,R_1+\delta]$
\begin{equation}
\gamma'(r) \geq \frac{1}{2r},
\end{equation}
as asserted.
\end{proof}

With the last lemma we can construct a lower bound $\tilde\gamma(r)$ for $\gamma(r)$ on the interval $[R_1,R_1+\delta]$, given by
\begin{equation}
\tilde \gamma(r) := \int_{R_1}^{r} \frac{1}{2s}\mathrm ds = \frac{1}{2}\ln\left(\frac{r}{R_1}\right).
\end{equation}
Furthermore, for each $\sigma_*\leq \delta$ we define
\begin{equation}
\gamma_* := \tilde \gamma(R_1+\sigma_*).
\end{equation}
We have
\begin{equation}\label{def_gamma_star}
\gamma(R_1+\sigma_*) \geq  \gamma_* = \frac 1 2\ln\left(\frac{R_1+\sigma_*}{R_1}\right).
\end{equation}

\begin{lemma} \label{lem_gamma_star}
Let
\begin{equation}
\sigma_* = R_1^{1+b} \qquad \mathrm{and} \qquad \Delta = c R_1^{1+d} 
\end{equation}
with the constants
\begin{equation} \label{def_const_abcd}
b=\frac{2\kappa+1}{\kappa (\kappa+1)}, \qquad c=\left(2^{2-\kappa}\pi C_\varrho^{\mathrm l} \right)^{-1},\qquad d=2-\kappa b.
\end{equation}
Define $r_2$ via
\begin{equation}
r_2 := \min\left\{2R_1,\max\{r \geq R_1+\sigma_*\,:\,\forall s\in[R_1+\sigma_*,r] \mathrm{\;we\,have\;} \gamma(s)\geq \gamma_*\}\right\}.
\end{equation}
Then, if $\tilde R_1>0$ is sufficiently small, we have for all $R_1\leq\tilde R_1$ that $r_2 <  R_1+\sigma_* + \Delta$.
\end{lemma}

Note that $\gamma_*$ is defined in (\ref{def_gamma_star}) and that $\sigma_* < \delta$ for $R_1$ sufficiently small since $b \geq \beta$. Furthermore note that the definition of $r_2$ is made such that $\gamma(r) \geq \gamma_*$ for $r\in [R_1 + \sigma_*, r_2]$.

\begin{proof}
We prove the statement by contradiction, i.e.~we assume $r_2\geq R_1+\sigma_*+\Delta$. This implies that $\gamma(r)\geq\gamma_*$ for all $r\in[R_1 + \sigma_*, R_1 + \sigma_* + \Delta]$. In this proof we employ $\Gamma(R_1)$ as a symbolic notation for a function with the property $\Gamma(R_1)\to 0$, as $R_1\to 0$. With this notation we have for all $r\in[R_1 + \sigma_*, R_1 + \sigma_* + \Delta]$
\begin{equation}
\gamma(r) \geq \gamma_* = \frac 1 2 \frac{\sigma_*}{R_1}(1-\Gamma(R_1)) = \frac{1}{2} R_1^{b} (1-\Gamma(R_1)).
\end{equation}
We use this observation to calculate
\begin{multline}
\frac{2 m_\lambda(R_1+\sigma_* + \Delta)}{R_1+\sigma_* + \Delta} = \frac{2}{R_1(1 + R_1^b + c R_1^d)} \bigg( 4\pi \int_{R_1}^{R_1+\sigma_* + \Delta} s^2 \varrho(s)\mathrm ds \\+ \frac 1 2 \int_{R_1}^{R_1+\sigma_* + \Delta} \frac{q^2(s)}{s^2}\mathrm ds \bigg).
\end{multline}
Using Lemma $\ref{lem_est_mat_c}$ and omitting the part of the first integral over the interval $(R_1,R_1+\sigma_*)$, as well as the whole second integral we obtain
\begin{align*}
\frac{2 m_\lambda(R_1+\sigma_* + \Delta)}{R_1+\sigma_* + \Delta} &\geq \frac{8\pi}{R_1} \left(1-\Gamma(R_1)\right) C_\varrho^{\mathrm l} \int_{R_1+\sigma_*}^{R_1+\sigma_*+\Delta} \frac{(\gamma(s))^{\kappa}}{s^2} \mathrm ds\\
&\geq 2^{3-\kappa}\pi  C_\varrho^{\mathrm l} R_1^{b\kappa-1} \left(1-\Gamma(R_1)\right) \frac{\Delta}{(R_1+\sigma_*)(R_1+\sigma_*+\Delta)} \\
&\geq 2^{3-\kappa}\pi  C_\varrho^{\mathrm l} R_1^{b\kappa-2+d} c  \left(1-\Gamma(R_1)\right) = 2(1-\Gamma(R_1)).
\end{align*}
In the last step the definitions of $c$ and $d$, given in equation (\ref{def_const_abcd}), have been used. Now, if $R_1$ is sufficiently small, we deduce $2m_\lambda(r) > r$. However, this would by the continuation criterion, Proposition \ref{prop_cont_crit}, imply that the solution $(\mu,\lambda,q)$ does not exist on the whole of the interval $[0,2R_1]$. Thus we have derived a contradiction to Corollary \ref{cor_cor51}.
\end{proof}

\begin{lemma} \label{lem_low_bd_ml}
Assume $r_2$ has been defined via the previous lemma. Then there holds:
\begin{equation}
\frac{2m_\lambda(r_2)}{r_2}\geq \frac 5 2.
\end{equation}
\end{lemma}

\begin{proof}
First note that the exponent $b$ defined in (\ref{def_const_abcd}) fulfills the conditions
\begin{equation} \label{cond_b}
%b\geq \beta,\qquad 
b\kappa>1,\qquad (\kappa+1)b>2.
\end{equation} 
Recall the Tolman-Oppenheimer-Volkov equation (\ref{tov_eq}) which reads
\begin{equation*}
p'(r) = -\mu'(r)(\varrho(r)+p(r))-\frac{2}{r}(p(r)-p_T(r))+\frac{q(r)q'(r)}{4\pi r^4}.
\end{equation*}
Using this equation one can derive the differential equation
\begin{multline}
\frac{\mathrm d}{\mathrm dr} \left[e^{(\mu+\lambda)(r)} \left(m_\lambda(r)+4\pi r^3 p(r) - \frac{q(r)^2}{2r} \right)\right] \\= e^{(\mu+\lambda)(r)}\left(4\pi r^2 \left(\varrho(r) + p(r) + 2p_T(r) \right)+\frac{q(r)^2}{r^2}\right).
\end{multline}
This in turn yields
\begin{multline}
e^{(\mu+\lambda)(r_2)}m_\lambda(r_2) = \int_{R_1}^{r_2} e^{(\mu+\lambda)(s)} \left(4\pi s^2\left(2\varrho(s)-z(s)\right)+\frac{q(s)^2}{s^2}\right) \mathrm ds \\
+ e^{(\mu+\lambda)(r_2)} \frac{q^2(r_2)}{2r_2} - 4\pi r_2^3 p(r_2) e^{(\mu+\lambda)(r_2)}.
\end{multline}
The variable $z(r)$ defined in Lemma \ref{lem_est_mat_c} has been substituted. The first Einstein equation (\ref{eieq1}) implies the relation
\begin{equation} \label{f_eeq_impl}
4\pi r e^{\lambda(r)} \varrho(r) = -\frac{\mathrm d}{\mathrm dr}\left(e^{-\lambda(r)}\right) + e^{\lambda(r)}\frac{m_\lambda(r)}{r^2} - 2\pi e^{\lambda(r)} \frac{q^2(r)}{r^3}.
\end{equation}
Furthermore we estimate using Lemma \ref{lem_est_mat_c} and $\gamma(r_2)=\gamma_*$ defined in (\ref{def_gamma_star})
\begin{align}
r_2p(r_2) &\leq C_p^{\mathrm u} \frac{\gamma(r_2)^{\kappa+1}}{r_2^3} \leq C_p^{\mathrm u} \frac{1}{(R_1+\sigma_*)^3}\left(\frac 1 2 \ln\left(\frac{R_1+\sigma_*}{R_1}\right)\right)^{\kappa+1} \label{est_rp}\\
&\leq \frac{C_p^{\mathrm u}}{2^{\kappa+1}}\frac{\sigma_*^{\kappa+1}}{R_1^{\kappa+1}\left(R_1+\sigma_*\right)^3} = \frac{C_\varrho^{\mathrm u}}{2^{\kappa+1}} \frac{R_1^{(\kappa+1)(1+b)}}{R_1^{\kappa+4}(1+R_1^b)^3} \leq CR_1^{(\kappa+1)b-3} \leq \frac{C}{R_1} \nonumber
\end{align}
where we have used the smallness of $R_1$ in the last step, as well as the second property in (\ref{cond_b}), and the general fact $\ln(1+x)\leq x$. Using (\ref{f_eeq_impl}) and (\ref{est_rp}) we obtain
\begin{align}
e^{(\mu+\lambda)(r_2)} m_\lambda(r_2) \geq& \int_{R_1}^{r_2} e^{\mu(s)} s\left(2-\frac{C_z^{\mathrm u}}{C_\varrho^{\mathrm l}} s^2\right) \left[-\frac{\mathrm d}{\mathrm ds}\left(e^{-\lambda(s)}\right)+\frac{m_\lambda(s)}{s^2}e^{\lambda(s)}\right] \mathrm ds \label{next_est} \\
&+ \int_{R_1}^{r_2} e^{(\mu+\lambda)(s)} \frac{q(s)^2}{s^2} \left(1-4\pi+2\pi \frac{C_z^{\mathrm u}}{C_\varrho^{\mathrm l}} s^2\right) \mathrm ds \nonumber \\
&+ e^{(\mu+\lambda)(r_2)}\frac{q^2(r_2)}{2r_2} -  4\pi C R_1^{(\kappa+1)b}. \nonumber
\end{align}
Note that for $R_1$ sufficiently small
\begin{equation}
\left(2-\frac{C_z^{\mathrm u}}{C_\varrho^{\mathrm l}} s^2\right) > 0\qquad\mathrm{and}\qquad 1-4\pi +2\pi \frac{C_z^{\mathrm u}}{C_\varrho^{\mathrm l}} s^2 < 0.
\end{equation}
Dropping the terms
\begin{equation*}
\int_{R_1}^{r_2} e^{(\mu+\lambda)(s)} \left(2-\frac{C_z^{\mathrm u}}{C_\varrho^{\mathrm l}} s^2\right) \frac{m_\lambda(s)}{s} \mathrm ds + e^{(\mu+\lambda)(r_2)}\frac{q^2(r_2)}{2r_2}
\end{equation*}
on the right hand side in (\ref{next_est}) we obtain
\begin{align}
e^{(\mu+\lambda)(r_2)} m_{\lambda}(r_2) \geq& -\left(2-\mathcal O(r_2^2)\right) R_1 \int_{R_1}^{r_2} e^{\mu(s)} \frac{\mathrm d}{\mathrm dr}\left(e^{-\lambda}\right) \mathrm ds \label{eq_for_int_w_mu} \\
&-(4\pi-1-\mathcal O(r_2^2)) \int_{R_1}^{r_2} e^{(\mu+\lambda)(s)} \frac{q(s)^2}{s^2}  \mathrm ds -\mathcal O\left( R_1^2\right). \nonumber
\end{align}
We use the notation $\mathcal O(r^k)$ for a general function such that there is a positive constant $M$ and a real number $r_0$ such that for all $r \leq r_0$ we have $\left|\mathcal O(r^k)\right| \leq M\left|r^k\right|$. Next we claim that
\begin{equation} \label{claim_mu}
\Delta_\mu(R_1):=\max_{r\in[R_1,r_2]} |\mu(r)-\mu(R_1)| \to 0,\qquad\mathrm{as}\quad R_1\to 0.
\end{equation}
We have
\begin{equation} \label{eq_prim_est_dmu}
\Delta_\mu(R_1) \leq \int_{R_1}^{r_2} e^{2\lambda(s)} \left| 4\pi s p(s) + \frac{m_\lambda(s)}{s^2} - \frac{q^2(s)}{2s^3}\right| \mathrm ds.
\end{equation}
We find an upper bound for this expression. Recall $e^{2\lambda(r)} \leq 20$ (cf.~equation  (\ref{eq_l_b_e2l_nnf})). For the term containing $p(r)$ we can use the estimate (\ref{est_rp}), i.e.~$4\pi s p(s) \leq \frac{C}{R_1}$. Using the inequality (\ref{a_b_ineq_charged}) and Proposition \ref{lem_qr_r2} we can write
\begin{align}
\frac{m_\lambda(s)}{s^2} &\leq \frac{m_g(s)}{s^2} \leq \frac{1}{s} \left(\frac 1 3 +\sqrt{\frac{1}{9} + \frac{q(s)^2}{3s^2}}\right)^2 \\
&\leq \frac{1}{s} \left(\frac{1}{3} + \sqrt{\frac 1 9 + \frac{C_q}{3} r_2^2}\right)^2 = \frac{1}{s} \left(\frac{4}{9} + \Gamma(r_2)\right) \leq \frac{1}{2s}, \nonumber
\end{align}
where we used smallness of $R_1$ in the last step and again the symbolic notation $\Gamma(r_2)$ for a function fulfilling $\Gamma(r_2)\to 0$ as $r_2\to 0$. Inserting this bound and the bound (\ref{est_rp}) into (\ref{eq_prim_est_dmu}), and using Proposition \ref{lem_qr_r2} for the term involving $q(s)$, we obtain
\begin{equation}
\Delta_\mu(R_1) \leq C \int_{R_1}^{r_2} \left( \frac{1}{R_1} + \frac{1}{2s} + \frac{s}{2}\right) \mathrm ds  \leq C\left(\frac{r_2-R_1}{R_1}+\ln\left(\frac{r_2}{R_1}\right) + r_2^2\right).
\end{equation}
Since $r_2/R_1\to1$, as $R_1\to 0$, (\ref{claim_mu}) follows. \par.
 A consequence of (\ref{claim_mu}) is that in the integral in (\ref{eq_for_int_w_mu}) we have
\begin{equation}
- e^{\mu(s)} \geq -e^{\mu(r_2)+\Delta_\mu(R_1)}.
\end{equation}
So we can cancel $e^{\mu(r_2)}$ on both sides of the inequality and obtain
\begin{align}
e^{\lambda(r_2)} m_\lambda(r_2) \geq& \left(2-\mathcal O(r_2^2)\right) e^{\Delta_\mu(R_1)} R_1 \left(1 -\sqrt{1-\frac{2m_\lambda(r_2)}{r_2}}\right) \label{est_delta_o}\\
&-(4\pi-1-\mathcal O(r_2^2)) e^{\Delta_\mu(R_1)} \int_{R_1}^{r_2} e^{\lambda(s)} \frac{q(s)^2}{s^2}  \mathrm ds -\mathcal O\left( R_1^{(\kappa+1)b}\right). \nonumber
\end{align}
Note that $r_2\leq R_1+\sigma_* + \Delta = R_1\left(1+ R_1^b+cR_1^d\right)\in\mathcal O(R_1)$.  Moreover we have
\begin{equation}
1-\sqrt{1-\frac{2m_\lambda(r)}{r}} = \frac{2m_\lambda(r)}{r\left(1+\sqrt{1-\frac{2m_\lambda(r)}{r}}\right)},
\end{equation}
and by Proposition \ref{lem_qr_r2} we also have
\begin{equation}
\int_{R_1}^{r_2} e^{\lambda(s)} \frac{q(s)^2}{s^2} \mathrm ds \to 0, \quad\mathrm{as}\; R_1\to 0.
\end{equation}
Moreover, since $\frac{R_1}{r}\to 1$ as $R_1\to 0$ for $r\in[R_1,2R_1]$, we can write
\begin{equation}
\left(2-\mathcal O\left(r_1^2\right) \right) e^{\Delta_\mu(R_1)}\frac{R_1}{r} = \left(2-\Gamma(R_1)\right)
\end{equation}
for $R_1$ sufficiently small, where $\Gamma(R_1)$ is the symbolic notation for a function satisfying $\Gamma(R_1) \leq 1$ and $\Gamma(R_1)\to 0$, as $R_1\to 0$. The actual function denoted by $\Gamma$ my change from line to line but the properties just mentioned stay preserved. So we obtain
\begin{equation}
e^{\lambda(r_2)}m_\lambda(r_2) \geq \left(2-\Gamma(R_1)\right) \frac{2m_\lambda(r_2)}{1+\sqrt{1-\frac{2m_\lambda(r_2)}{r_2}}} + \Gamma(R_1), 
\end{equation}
This is equivalent to
\begin{align*}
&1 \geq \left( 4-\Gamma(R_1) \right)\sqrt{1-\frac{2m_\lambda(r_2)}{r_2}} \frac{1}{\left(1+\sqrt{1-\frac{2m_\lambda(r_2)}{r_2}}\right)} + \Gamma(R_1) \\
\Leftrightarrow \qquad & -(3-\Gamma(R_1)) \sqrt{1-\frac{2 m_\lambda(r_2)}{r_2}}  \geq - (1-\Gamma(R_1)) \\
\Leftrightarrow \qquad & \frac{2m_\lambda(r_2)}{r_2} \geq \frac{8}{9}-\Gamma(R_1).
\end{align*}
We see that $2m_\lambda(r_2)/r_2\to\frac 89$, as $R_1\to 0$. So in particular there is $\tilde R_1$ such that for all $R_1\leq\tilde R_1$ we have
\begin{equation}
\frac{2m_\lambda(r_2)}{r_2}\geq \frac{4}{5},
\end{equation}
as asserted.
\end{proof}

\begin{prop}
Define the variable
\begin{equation} \label{def_mak_var}
x(r):=\frac{m_\lambda(r)}{r\gamma(r)}.
\end{equation}
$x(r)$ diverges at a finite radius $R_2\leq 2R_1$ if $\mu_c$ is chosen sufficiently small. Moreover $R_2/R_1\to 1$, as $R_1\to 0$.
\end{prop}
\begin{proof}
We have
\begin{multline}
rx'(r) = \frac{4\pi r^2 \varrho(r)}{\gamma(r)} + \frac{q(r)^2}{2r^2\gamma(r)} - x(r) - \frac{x(r)}{\gamma(r)}\frac{q_0 e^{(\lambda+\mu)(r)}}{1+I_q(r)} \frac{q(r)}{r} \\+ \frac{x(r)}{\gamma(r)(1-2x(r)\gamma(r))}\left(4\pi r^2 p(r) + x(r)\gamma(r) -\frac{q^2(r)}{2r^2}\right)-\frac{x(r)}{\gamma(r)}\frac{L_0}{r^2+L_0}.
\end{multline}
Since it is the aim now to find a lower bound for the expression $rx'(r)$ we can drop the positive terms
\begin{equation*}
\frac{4\pi r^2\varrho(r)}{\gamma(r)} + \frac{q(r)^2}{2r^2 \gamma(r)} + \frac{4\pi r^2 p(r) x(r)}{\gamma(r)(1-2x(r)\gamma(r))}.
\end{equation*}
Then, using the estimate
\begin{equation*}
\frac{\mathrm d}{\mathrm dr} \ln\left(1+I_q(r)\right) \leq \frac{q_0}{\sqrt{r^2+L_0}} e^{\lambda(r)} \frac{q(r)}{r}
\end{equation*}
(cf.~equation (\ref{est_ddr_1piq})) we obtain
\begin{multline} \label{est_m_var_red}
rx'(r) \geq \frac{x^2(r)}{1-2x(r)\gamma(r)}-x(r) - \frac{x(r)}{\gamma(r)} \left( \frac{q_0}{\sqrt{r^2+L_0}}e^{\lambda(r)} q(r) + \frac{L_0}{r^2+L_0}\right) \\- \frac{q(r)^2}{2r^2\gamma(r)} \frac{x(r)}{1-2x(r)\gamma(r)}.
\end{multline}
Note that by Lemma \ref{lem_qr_r2} the factor in the third term can be estimated by
\begin{equation}
\frac{q_0}{\sqrt{r^2+L_0}} e^{\lambda(r)}q(r) + \frac{L_0}{r^2+L_0} \leq 1+\Gamma(R_1),
\end{equation}
where $\Gamma(R_1)$ is the symbolic notation for a function fulfilling $\Gamma(R_1)\to 0$, as $R_1\to 0$. Using this observation and reordering terms we write (\ref{est_m_var_red}) as
\begin{multline} \label{est_m_var_red_2}
rx'(r) \geq x(r) \left(\frac 3 5 \frac{x(r)}{1-2x(r)\gamma(r)} - \frac{1+\Gamma(R_1)}{\gamma(r)}\right) \\
+\left(\frac{x(r)}{5}-\frac{q(r)^2}{2r^2\gamma(r)} \right) \frac{x(r)}{1-2x(r)\gamma(r)} + \frac 1 5 \frac{x^2(r)}{1-2x(r)\gamma(r)}-x(r)
\end{multline}
to prepare for the next estimates.\par
In the remainder of this proof let $r\in[r_2,\frac{76}{75}r_2]$. By Lemma \ref{lem_low_bd_ml} for $r>r_2$ we have
\begin{equation} \label{eq_mr_r2r}
\frac{2m_\lambda(r)}{r} \geq \frac 4 5 \frac{r_2}{r} \geq \frac{15}{19}.
\end{equation}
Then we have
\begin{equation}
\frac 3 5\frac{x(r)}{1-2x(r)\gamma(r)} = \frac{3}{5\gamma(r)} \frac{m_\lambda(r)}{r} \frac{1}{1-\frac{2m_\lambda(r)}{r}} \geq \frac{1}{\gamma(r)} \frac{9}{38} \frac{19}{4} = \frac{1}{\gamma(r)} \frac{9}{8} \geq \frac{1+\Gamma(R_1)}{\gamma(r)}.%\frac{2r_2}{5r-4r_2} \to \frac{2}{\gamma(r)},\quad\mathrm{as}\;r\to r_2.
\end{equation}
So the first term in (\ref{est_m_var_red_2}) is positive and can be dropped if $R_1$ is sufficiently small. In order to deal with the second term we estimate using Lemma \ref{lem_qr_r2} and (\ref{eq_mr_r2r})
\begin{align}
\frac{x(r)}{5} - \frac{q(r)^2}{2r^2\gamma(r)} &\geq \frac{x(r)}{5}-\frac{C_q^2r^2}{2\gamma(r)} = \frac{1}{\gamma(r)}\left(\frac{m_\lambda(r)}{5r}-\frac{C_q^2r^2}{2}\right) \geq \frac{1}{\gamma(r)} \left(\frac{152}{1875}\frac{r_2}{r}-\frac{C_q^2r^2}{2}\right) \\
&= \frac{C_q^2r^2}{2\gamma(r)} \left(C\frac{r_2}{r^3}-1\right) \leq \frac{C_q^2 r_2^2}{2\gamma(r)} \left(\frac{C}{r^2}-1\right). \nonumber
\end{align}
This is positive if $R_1$ is chosen sufficiently small since $r_2 \to 0$, as $R_1\to 0$. So the second term can be dropped as well in (\ref{est_m_var_red_2}). We are left with
\begin{equation}
rx'(r) \geq \frac 1 5 \frac{x^2(r)}{1-2x(r)\gamma(r)} - x(r) \geq \frac{1}{5-4\frac{r_2}{r}} x^2(r) - x(r) \geq \frac{19}{20} x^2(r) - x(r).
\end{equation}
We solve this differential inequality with the method of separation of variables. We have
\begin{equation}
\int_{r_2}^r\frac{1}{s}\mathrm ds = \int_{x(r_2)}^{x(r)} \frac{1}{\frac{19}{20} x^2-x}\mathrm dx
\end{equation}
which yields
\begin{equation}
x(r) \geq \frac{20x(r_2)}{(20-19x(r_2))\frac{r}{r_2}+19x(r_2)}.
\end{equation}
For the next steps we note the lower bound 
\begin{equation}
x(r_2) \geq \frac{2}{5C_\gamma} r_2^{-\frac{2}{\kappa+1}},
\end{equation}
which is a combination of Lemma \ref{lem_c_gamma} and Lemma \ref{lem_low_bd_ml}. This yields
\begin{equation}
x(r) \geq \frac{20}{19 + r\left(50 C_\gamma r_2^{\frac{2}{\kappa+1}-1}-\frac{19}{r_2}\right)}.
\end{equation}
We note that the denominator is positive if $r=r_2$. However, if we substitute 
\begin{equation}
r=\tilde R_2:=r_2\left(1+\frac{50C_\gamma}{19}r_2^{\frac{1}{\kappa+1}}\right)
\end{equation}
the denominator becomes
$$
50C_\gamma r_2^{\frac{1}{\kappa+1}}\left(r_2^{\frac{1}{\kappa+1}}-1\right) + \frac{2500 C_\gamma^2}{19} r_2^{\frac{3}{\kappa+1}}
$$
which is less than zero if $r_2$ is small enough. By continuity the denominator has at least one zero in the interval $[r_2,\tilde R_2]$. We call the first zero $R_2$.  By construction of the variable $x(r)$ in (\ref{def_mak_var}) we have $\gamma(R_2)=0$ and since $r_2\leq R_2 \leq \tilde R_2$ we have $R_2/R_1\to 1$ as $R_1\to 0$, since already $r_2/R_1\to 1$, as $R_1\to 0$, as implied by Lemma \ref{lem_gamma_star}. 
\end{proof}

\end{document}